\documentclass[sigconf,nonacm]{aamas}
\usepackage{balance}
\usepackage{booktabs}
\graphicspath{{./}{../figs/}{../../figs/}}
\newtheorem{assumption}{Assumption}
\newtheorem{remark}{Remark}
\newcommand{\bs}{\mathbf{s}}
\newcommand{\be}{\mathbf{e}}
\newcommand{\bmu}{\boldsymbol{\mu}}
\newcommand{\R}{\mathbb{R}}
\newcommand{\N}{\mathcal{N}}
\newcommand{\Ncal}{\mathcal{N}}

\newcommand{\E}{\mathbb{E}}
\title{Beyond Independent Values: Optimal Auctions with Correlated Valuations for LLM-Native Advertising}
\newcommand{\1}{\mathbf{1}}
\theoremstyle{acmplain}
\newtheorem{stheorem}{Theorem}
\newtheorem{slemma}{Lemma}
\newtheorem{sproposition}{Proposition}
\theoremstyle{acmdefinition}
\newtheorem{sremark}{Remark}
\author{Shiman Sheng}
\affiliation{\institution{The University of Adelaide}\city{Adelaide}\country{Australia}}
\email{shiman.sheng@adelaide.edu.au}
\author{Mingyu Guo}
\affiliation{\institution{The University of Adelaide}\city{Adelaide}\country{Australia}}
\email{mingyu.guo@adelaide.edu.au}
\begin{document}
\pagestyle{plain}

\begin{abstract}
Conversational assistants are becoming the place where people ask what to buy, and sponsored content is moving into their answers. Several advertisers then compete for one sponsored slot, and all of them read the same conversation, so their willingness to pay usually moves together. Existing auctions for LLM advertising treat these values as independent. We therefore present \textsc{AffMyerson}, an auction that captures this correlation and turns it into revenue. We model the advertisers' values given the conversation as a multivariate normal, the standard model for correlated quantities, and turn the economic premise of positive dependence (affiliation) into a sign condition on the precision matrix. We design a neural parameterization whose every output satisfies it. We prove that every prior in this family yields a truthful, ex-post individually rational auction, so estimation error can cost revenue but never truthfulness. We derive the conditional virtual values in closed form. By the design of the prior, pricing an auction needs no matrix inversion and takes $O(n^2)$ time, fast enough to run at query time. We also work out the log-normal family used in empirical ad-auction work. There we prove that the auction is truthful for every prior and revenue-optimal when the prior is correct, and that the virtual value is regular everywhere whenever the conditional spread is below a loose threshold, which our fitted priors stay well under. With LLM-agent bidders, \textsc{AffMyerson} earns $6.2\%$ more than a second-price auction with a tuned reserve, $6.1\%$ more than Ronen's lookahead auction with the same prior, and $9.2\%$ more than Myerson's auction run as if values were independent. This gain is still $4.5\%$ on real keyword search queries where willingness to pay is hard to read. The correlation comes from common demand. Agents that judge the conversation bid together and agents told to calculate do not. The gain from conditioning moves with the correlation.
\end{abstract}
\keywords{LLM advertising; optimal auction design; correlated values; affiliation; M-matrix priors; truthful mechanisms}

\maketitle
\thispagestyle{plain}
\section{Introduction}

Conversational assistants such as ChatGPT, Perplexity, Google AI Overviews, and Bing Copilot are becoming the place where people ask what to buy, and sponsored suggestions are moving into their answers~\cite{hajiaghayi2024,dutting2024}. Consider a user who asks for ``the best hiking boots for a trip to the Alps in October.'' Several outdoor brands would pay to have their boot mentioned in the answer. All of them read the same conversation, so they can all see that this user is ready to buy and that October is peak season for the category. If the conversation is worth more than usual to one brand, it is likely to be worth more to the others as well. Milgrom and Weber's linkage principle says that such \emph{correlation} is worth money to the platform~\cite{milgrom1982}. A rival's high bid is evidence that the winner's own value is high, so the platform can charge the winner more without losing the sale. The auctions proposed so far for LLM advertising treat advertisers' values as independent, as in Myerson's independent-private-values (IPV) model~\cite{myerson1981}. They include segment auctions in retrieval-augmented generation~\cite{hajiaghayi2024}, token auctions inside decoding~\cite{dutting2024}, alignment-based allocation~\cite{llmauction2025}, and Myersonian payments in two-stage pipelines~\cite{lera2026,quality2026}. None of them uses what one advertiser's bid says about another's value, and that is the information this paper builds into the auction.

\begin{figure*}[t]\centering
\includegraphics[width=\textwidth]{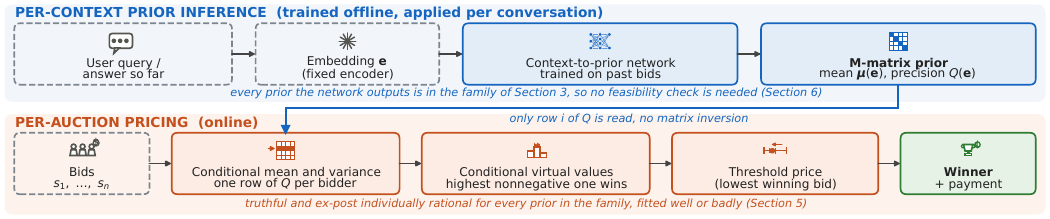}
\Description{Two-stage pipeline: a network maps the conversation to a joint prior, and the auction prices each opportunity from that prior.}
\caption{Two stages. \textbf{Per conversation} (top): a network trained offline on past bids maps the conversation to a joint prior over the advertisers' values. Every prior it can output is in the family of Section~\ref{sec:model}, so the auction below is truthful whatever the fit (Theorem~\ref{thm:robust}). \textbf{Per auction} (bottom): the winner is the bidder with the highest conditional virtual value and pays the lowest bid at which it would still win (Section~\ref{sec:mech}).}\label{fig:pipeline}
\end{figure*}

\paragraph{The prior the platform needs, and how we build it.}
The prior must do three things: carry correlation among $n$ values, be fast to fit and to condition on the rivals' bids at auction time, and keep the auction truthful when the estimate is inaccurate, since conversations vary too much for any fitted prior to be exact. We use the multivariate normal, the standard model for several correlated quantities, and build the economic premise of positive dependence (\emph{affiliation}) into it through one condition on its shape: the \emph{precision matrix} (the inverse covariance) has nonpositive off-diagonal entries. Statisticians call this the M-matrix condition, and for a Gaussian this one sign condition is all that affiliation requires~\cite{karlin1980}. It also lets the platform read one bidder's conditional distribution from a single row of the matrix, with no inversion. The auction is then the natural one: rank bidders by their conditional virtual value (Myerson's virtual value~\cite{myerson1981} computed under the distribution of the bidder's value given the rivals' bids), give the slot to the top-ranked bidder if its virtual value is nonnegative, and charge it the lowest bid at which it would still win. Under the true prior this is the revenue-optimal truthful auction~\cite{rtc2016}. Truthfulness does not need the true prior. For every prior in the family, fitted well or badly, the auction is dominant-strategy truthful and never leaves a bidder worse off (Theorem~\ref{thm:robust}), and revenue degrades continuously in the estimation error (Proposition~\ref{prop:stab}). A learned prior is therefore safe to deploy, since the only thing a bad estimate can cost is revenue.

A second prior family meets the same three requirements. Advertising values are commonly modeled as log-normal, in distributions fitted to bid data~\cite{lahaie2007}, in simulation~\cite{edelman2010}, and in field experiments~\cite{ostrovsky2011}, so we work out that family as well. The auction stays truthful and optimal there, and its virtual value is regular everywhere (so no ironing is needed) whenever the conditional spread of log values is below a threshold we compute, $\sigma^\star=1.5176$. The threshold is loose in practice: at $\sigma^\star$, one standard deviation in log scale multiplies a value by $4.6$, our fitted priors never exceed $0.64$, and every prior the parameterization of Section~\ref{sec:neural} outputs satisfies it by construction (Appendix~S4).

The fitting method is chosen the same way. We use a small network that reads a sentence embedding of the conversation and can only output priors inside the family, so nothing has to be checked at auction time. On our designed queries, the network can draw on the text of the conversation. On the short search queries we draw from Amazon ESCI and MS MARCO, a few words say little about how much an advertiser values the query, so the network has nothing to learn from the text, and a convex maximum-likelihood fit with no network earns the same revenue. Either way the auction stays truthful, so the platform can choose the fitting method by what its data support.

\paragraph{Contributions.}
\begin{enumerate}
  \item \sloppy\textbf{Two tractable prior families with incentive and optimality guarantees.} Roughgarden and Talgam-Cohen~\cite{rtc2016} verify the conditions of their optimal auction for correlated bidders on one symmetric normal. We verify them for every normal prior with an M-matrix precision, for any means and any number of bidders, which needs cross-monotonicity through the sign of the precision matrix and a revenue identity on the unbounded support (Theorem~\ref{thm:opt}). For the log-normal family we give an exact regularity condition: the virtual value is globally regular if and only if the conditional log-scale is at most $\sigma^\star=1.5176$, a threshold that is loose in practice (our fitted priors stay below $0.64$) and that the parameterization keeps by design (Theorem~\ref{thm:loglink}).
  \item \textbf{Truthful whatever the fit.} For every prior in the family, accurate or not, truthful bidding is a dominant strategy and no bidder ever pays more than its value (Theorem~\ref{thm:robust}). Learned mechanisms such as CAN~\cite{huo2025can} also get truthfulness from their structure. Here it comes with exact optimality when the prior is correct, and estimation error costs only revenue, by an amount we bound in the size of the error (Proposition~\ref{prop:stab}).
  \item \textbf{Fast enough to run inside a conversation.} Because we parameterize the prior by its precision matrix rather than its covariance, pricing an auction needs no matrix inversion and takes $O(n^2)$ time, microseconds in practice (Section~\ref{sec:mech}).
  \item \textbf{How LLM agents turn common demand into correlated bids.} When LLM agents bid on behalf of advertisers, given a brand profile, the conversation, and a demand signal shared by all brands, they reproduce two features of real sponsored-search markets: bids rise with purchase intent, and rival bids move together. The shared signal does not make the correlation automatic. Across seven prompt schemes and two models, agents that judge the conversation bid together (within-query correlation $+0.78$), agents told to calculate do not ($+0.06$), and the gain from conditioning rises and falls with the correlation. To our knowledge, this is the first study of how LLM bidding agents turn a common demand shock into correlated bids, a setting we expect to become common (Section~\ref{sec:exp-ablation}).
  \item \textbf{Revenue.} With LLM-agent advertisers and hidden demand shocks, \textsc{AffMyerson} earns $6.2\%$ more than a second-price auction with a tuned reserve, $6.1\%$ more than Ronen's lookahead auction~\cite{ronen2001} with the same prior, $9.2\%$ more than Myerson's auction run as if values were independent, and $1.1\%$ more than CAN, a neural mechanism trained on the same bids~\cite{huo2025can}. We call the gap over the unconditioned Myerson auction the \emph{affiliation rent}. The rent vanishes when correlation is removed from the data, grows monotonically as it is added back, survives twelve deliberate corruptions of the prior, and is still $+4.5\%$ on $580$ real keyword search queries where willingness to pay is hard to read (Section~\ref{sec:exp}). In a market where the conversation matters to one advertiser and not to the other, \textsc{AffMyerson} earns $24.8\%$ more than CAN, and the whole gain comes from the part of our design that reads the conversation (Section~\ref{sec:exp-rev}).
\end{enumerate}
\section{Related Work}

\paragraph{LLM-native advertising.}
Segment auctions~\cite{hajiaghayi2024} auction discourse segments in a retrieval-augmented generation (RAG) pipeline; token auctions~\cite{dutting2024} embed second-price logic into decoding; \textsc{LLM-Auction}~\cite{llmauction2025} casts ad allocation as preference alignment; LERA~\cite{lera2026} and quality-preserving mechanisms~\cite{quality2026} apply critical-value (Myersonian) payments in two-stage RAG systems. None of them models the correlation among advertisers' values, so none of them can price it.

\paragraph{Optimal auction theory, affiliation, and M-matrices.}
Myerson~\cite{myerson1981}, Milgrom and Weber~\cite{milgrom1982}, and Roughgarden and Talgam-Cohen~\cite{rtc2016} are the theory we build on. Our mechanism is the Gaussian case of the latter's Theorem 5.3, and Section~\ref{sec:theory} states the exact relation, including to Cr\'emer and McLean~\cite{cremer1988}. For general correlated priors, the optimal deterministic auction is NP-hard to compute with three or more bidders~\cite{pp2011}, and the optimal randomized auction is a linear program whose size grows with the support~\cite{dfk2011}, which is one reason to work inside a family where it has a closed form. Two lines of work learn the auction itself rather than the prior. For a given correlated prior, Guo et al.~\cite{guo2024nninterp} interpolate Myerson's greedy rule, which need not be monotone, and certify monotonicity afterwards by integer programming. In our family the greedy rule is already monotone (Lemmas~\ref{lem:reg} and~\ref{lem:cross}). Huo et al.~\cite{huo2025can} learn rank-score functions whose architecture enforces the single-crossing conditions, so truthfulness holds without any prior, and train them on about $100{,}000$ samples of one market. Both learn the allocation, and neither learns the value distribution. We learn the prior instead and keep the auction in closed form, so the platform gets a distribution it can inspect and that varies with the conversation, exact optimality within the family, and pricing in $O(n^2)$ time. What we give up is a distributional assumption, which Section~\ref{sec:exp-rev} tests. Learning the prior from samples is studied in~\cite{cole2014} and, for correlated buyers, in~\cite{fu2014}, and~\cite{ostrovsky2011} gives field evidence for prior-based reserves. To our knowledge, we are the first to use Karlin and Rinott's characterization of affiliated Gaussians by an M-matrix precision~\cite{karlin1980} in computational advertising or mechanism design. It is standard in graphical models~\cite{lauritzen1996}, and estimation under it is well studied: the constrained Gaussian maximum-likelihood fit is convex and has a solution under mild conditions~\cite{luz2019,slawski2015}, which Section~\ref{sec:exp-realq} uses.
\section{Model and Preliminaries}\label{sec:model}

This section sets up the environment, states the prior, and derives the two things the mechanism needs from it: how to condition one bidder's value on the rivals' bids, and why the prior's shape guarantees affiliation.

\subsection{Environment}

An LLM system is generating an answer to a user query $x$. At certain positions in the generation process the system identifies an \emph{ad opportunity}: a location where sponsored content can be inserted in a contextually relevant way. For each opportunity, a set $\Ncal=\{1,\dots,n\}$ of candidate advertisers competes for a single slot.

Each opportunity is a single-item auction with no repeated interaction. Each advertiser $i$ holds a private valuation signal $s_i\in\R$, its willingness to pay for appearing at this position for query $x$, and its value is its own signal, $v_i=s_i$. Three settings differ in how the $s_i$ relate. Under \emph{independent private values} (IPV) the signals are drawn independently. Under \emph{correlated private values} (our setting, the ``correlated values'' cell of~\cite{rtc2016}) each bidder still knows its own value, but the signals are jointly distributed, so a rival's report tells the seller something about $i$'s value. Under \emph{interdependent values} a bidder's value itself depends on others' signals, which we do not treat. We assume the dependence is positive (\emph{affiliation}, defined below). Negative correlation, as between close substitutes, is outside the model (Section~\ref{sec:conclusion}).

Two facts about sponsored-search values shape the prior. First, values depend on the query, since a conversation with strong commercial intent is worth more to every advertiser than a navigational one, so the prior's mean must move with the context. Second, values move together, because a demand shock the platform does not observe, such as a product launch or a seasonal rise in demand, raises every advertiser's value at once, and this co-movement is affiliation. We capture both with one model, a multivariate normal whose mean is a function of the context and whose precision matrix (the inverse covariance) is restricted to the class that affiliation requires.

Three terms recur below. \emph{Affiliation} is the economic premise: when one advertiser's value is high, every other advertiser's value is more likely to be high too. Its formal definition is $\mathrm{MTP}_2$, that the joint density is log-supermodular. An \emph{M-matrix} is a positive definite matrix whose off-diagonal entries are nonpositive, and it is what $\mathrm{MTP}_2$ reduces to for a Gaussian: the precision matrix must be an M-matrix, which says that every pair of bidders is nonnegatively partially correlated (Proposition~\ref{prop:mtp2}).

\begin{assumption}[Affiliated valuations]\label{as:mvn}
Given a context embedding $\be$ (derived from the query and the response generated so far), the valuation vector $\bs=(s_1,\dots,s_n)$ is drawn from a multivariate normal (MVN) distribution
\begin{equation}
\bs\mid\be \sim \N\big(\bmu(\be),\,Q(\be)^{-1}\big),
\end{equation}
where the precision matrix $Q(\be)\in\R^{n\times n}$ satisfies (i) $Q\succ 0$ (symmetric positive definite) and (ii) $Q_{ij}\le 0$ for all $i\neq j$. A matrix satisfying (i)--(ii) is a \emph{symmetric positive definite M-matrix}.
\end{assumption}

Condition (ii) is the one that matters, since for a Gaussian it is equivalent to affiliation (Proposition~\ref{prop:mtp2}). Section~\ref{sec:whymvn} explains why we choose this distribution, after the quantities it must provide have been introduced. Section~\ref{sec:neural} describes how we parameterize and fit this class.

\begin{remark}[Negative signals]\label{as:pos}
A Gaussian puts some mass on negative values, and so do our fitted priors (Appendix~S1). This has no formal consequence, since the incentive results hold on all of $\R^n$, a bidder with a negative signal has a negative conditional virtual value (it is always below the signal itself, Definition~\ref{def:cvv}) and so is never allocated and pays zero, and every signal in our data is an observed bid in cents.
\end{remark}

\begin{proposition}[Affiliation $\Leftrightarrow$ M-matrix]\label{prop:mtp2}
A multivariate normal $\N(\bmu,Q^{-1})$ is affiliated ($\mathrm{MTP}_2$) if and only if $Q\succ0$ and $Q_{ij}\le 0$ for every pair $i\neq j$.
\end{proposition}
The characterization is Karlin and Rinott's~\cite{karlin1980}.

\emph{Why affiliation pays.} Under IPV the winner's payment depends on competitors' bids only through the ranking. Under affiliation, a high bid from one advertiser is evidence that the others' values are higher than usual as well, so when a rival bids high the seller can raise its estimate of the winner's value and \emph{link} the winner's payment to the rivals' reports. The winner's own value does not change, but what the seller knows about it does, and this is what reduces the winner's information rent, the surplus a bidder keeps because the seller does not know its value (the linkage principle~\cite{milgrom1982}).

\subsection{Conditioning from the Precision Matrix}

At auction time the mechanism needs, for each bidder $i$, the mean and standard deviation of $s_i$ given the rivals' reports $\bs_{-i}$ (here and below, $-i$ denotes all bidders other than $i$, and a subscript $i\mid-i$ a quantity of bidder $i$ conditional on them), which we call the \emph{conditional moments}. Parameterizing the precision matrix $Q$ directly, rather than the covariance $\Sigma=Q^{-1}$, means they can be computed without inverting a matrix:

\begin{lemma}[Precision-based conditioning]\label{lem:cond}
For $\N(\bmu,Q^{-1})$, the conditional law of $s_i$ given $\bs_{-i}$ is
\begin{equation}
\begin{aligned}
s_i\mid\bs_{-i} &\sim \N\big(\mu_{i\mid-i},\,\sigma_{i\mid-i}^2\big),\qquad \sigma_{i\mid-i}^2 = \tfrac{1}{Q_{ii}},\\
\mu_{i\mid-i} &= \mu_i-\tfrac{1}{Q_{ii}}\textstyle\sum_{j\neq i}Q_{ij}(s_j-\mu_j).
\end{aligned}
\end{equation}
\end{lemma}

The lemma is standard in Gaussian graphical models~\cite{lauritzen1996}. Both conditional moments come from a single row of $Q$, so the auction in Section~\ref{sec:mech} is fast to price.

\subsection{Why a Context-Dependent Gaussian?}\label{sec:whymvn}

Our prior is a Gaussian whose mean moves with the conversation and whose precision matrix is an M-matrix. We choose it for three reasons.
\emph{(1) Theoretical support.} Roughgarden and Talgam-Cohen~\cite{rtc2016} verify the conditions of their ex-post optimality framework for one example, a symmetric multivariate normal. We verify them for the whole M-matrix class, with arbitrary context-dependent means and covariances (Proposition~\ref{prop:mtp2}, Section~\ref{sec:theory}), so the optimal auction for correlated values extends from one example to the whole class, and through a log transformation to the log-normal family (Section~\ref{sec:robust}).
\emph{(2) Everything the auction needs has a closed form.} Per-query pricing needs three things: the conditional distribution (Lemma~\ref{lem:cond}), because the price is set after the rivals report; the conditional virtual value (Section~\ref{sec:mech}), because the optimal allocation ranks bidders by it; and a checkable criterion for affiliation (Proposition~\ref{prop:mtp2}), because the incentive results need it to hold for every prior the fitting procedure can output. The Gaussian gives all three in closed form. The first two are one-line formulas in the Mills ratio, and the third is a sign condition on the precision matrix that every output of the parameterization satisfies (Section~\ref{sec:neural}).
\emph{(3) Truthful for every member of the family.} The fitted prior is an estimate, and real bids follow the Gaussian only approximately. Because every prior in the family gives a truthful auction (Theorem~\ref{thm:robust}), estimation error can cost revenue (Proposition~\ref{prop:stab}) but not incentives.
\section{The \textsc{AffMyerson} Mechanism}\label{sec:mech}

\begin{figure*}[t]\centering
\includegraphics[width=0.60\textwidth]{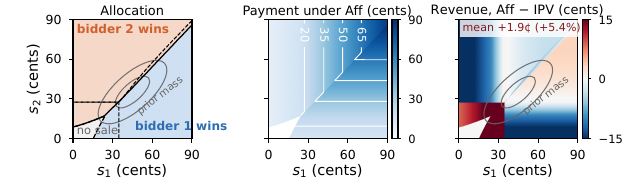}
\Description{Two-bidder example: the conditional distribution, the winner's payment surface, and the revenue gain over IPV-Myerson.}
\caption{How conditioning moves money (illustrative two-bidder prior with means $45$ and $35$ cents, standard deviations $13$ and $12$, $\rho=0.7$; derived quantities Appendix~S1). \textbf{Left:} allocation regions of \textsc{AffMyerson} (solid) vs.\ IPV-Myerson (dashed). The winner's threshold moves with the opponent's report. \textbf{Center:} the winner's payment $t_0$, which rises with the opponent's bid ($\partial t_0/\partial s_j>0$ in this bivariate case, Section~\ref{sec:compare}). \textbf{Right:} the revenue difference at each bid pair (\textsc{AffMyerson} minus IPV-Myerson). The ellipses enclose $1\sigma$ and $2\sigma$ of the prior, that is, where most bid pairs fall. Averaged over the prior the difference is $+1.9$ cents ($+5.4\%$).}\label{fig:mechanism}
\end{figure*}

We now define the mechanism, first the quantity it ranks, then the allocation and payment rule, and finally its comparison with IPV-Myerson (Section~\ref{sec:compare}).

\subsection{Conditional Virtual Values}

Myerson's optimal auction ranks bidders by virtual value, the bidder's value minus the information rent it can extract, $s_i-(1-F(s_i))/f(s_i)$~\cite{myerson1981}. Milgrom and Weber~\cite{milgrom1982} and Roughgarden and Talgam-Cohen~\cite{rtc2016} extend the notion to correlated values by replacing the marginal distribution with the distribution of $s_i$ conditional on the other bids. Under our prior that conditional distribution is normal with the moments of Lemma~\ref{lem:cond}, and the virtual value has a closed form in one standard function, the Mills ratio $R(z)=(1-\Phi(z))/\phi(z)$, which is the tail probability of a standard normal divided by its density. $R$ is decreasing and tabulated, so the virtual value below costs one function evaluation per bidder.

\begin{definition}[Conditional virtual value]\label{def:cvv}
$\psi_i(s_i\mid\bs_{-i}) = s_i-\sigma_{i\mid-i}\,R(z_i)$, where $z_i=(s_i-\mu_{i\mid-i})/\sigma_{i\mid-i}$ is the standardized residual.
\end{definition}

\subsection{Allocation and Payment}

The mechanism allocates to the bidder with the highest \emph{nonnegative} conditional virtual value and charges the \emph{threshold bid}, the lowest report at which the winner still wins. The threshold $t_0$ solves the implicit equation
\begin{equation}\label{eq:threshold}
\psi_{i^*}(t_0\mid\bs_{-i^*})=\max\big(0,\ \max_{j\neq i^*}\psi_j(s_j\mid t_0,\bs_{-\{i^*,j\}})\big),
\end{equation}
and the payment is $p_{i^*}=t_0$. Since $v_{i^*}=s_{i^*}$, the critical payment equals the threshold signal, and Lemma~\ref{thm:threshold} shows that $t_0$ is unique in $(0,s_{i^*}]$, so we find it by bisection on that bracket (numerical tolerance: Appendix~S2). If all $\psi_i<0$ the item is unallocated. Each virtual value needs one row of $Q$, so pricing costs $O(n^2)$ plus a constant number of $O(n)$ bisection steps, microseconds in pure Python (Appendix~S14).

\subsection{Comparison with IPV-Myerson: the Affiliation Rent}\label{sec:compare}

Under independent private values, Myerson's virtual value uses the \emph{marginal} distribution of $s_i$. \textsc{AffMyerson} uses the conditional one, which differs in two ways. Its spread is smaller, since the rivals' bids resolve part of the uncertainty about $s_i$ (the law of total variance), and a smaller spread leaves the bidder a smaller information rent $\sigma_{i\mid-i}R(z_i)$. Its mean moves with the rivals' bids, so a high rival bid raises the threshold $t_0$ the winner must clear ($\partial t_0/\partial s_j>0$ for $n{=}2$; Figure~\ref{fig:mechanism}). In the two-brand prior of Figure~\ref{fig:mechanism}, a rival bid of $61$ cents moves bidder 1's conditional mean from $45$ to $64.7$ cents and shrinks its conditional deviation from $13$ to $9.3$. For $n\ge3$ a high rival bid also lowers the other rivals' virtual values, so the net effect on $t_0$ can take either sign (Appendix~S2). Truthfulness does not depend on this, only on monotonicity in a bidder's own report (Lemmas~\ref{lem:reg} and~\ref{lem:cross}). We call the revenue gap between the two the \emph{affiliation rent}. Whether it survives an estimated prior is the question of the next section.
\section{Theoretical Analysis}\label{sec:theory}

We first establish the two monotonicity properties the incentive argument needs, then the incentive and optimality results, and finally what survives when the prior is only an estimate.

\subsection{Regularity and Cross-Monotonicity}

\begin{lemma}[Unconditional regularity]\label{lem:reg}
For every $\bs_{-i}$, $\psi_i(s_i\mid\bs_{-i})$ is strictly increasing in $s_i$.
\end{lemma}

The derivative exceeds one, so regularity holds and no ironing (Myerson's repair for non-monotone virtual values) is needed~\cite{bagnoli2005}. Because the conditionals stay Gaussian with a variance that does not depend on $\bs_{-i}$, regularity holds at every $\bs_{-i}$.

\begin{lemma}[Cross-monotonicity]\label{lem:cross}
For all $j\neq i$, $\psi_j(s_j\mid\bs_{-j})$ is nonincreasing in $s_i$, strictly decreasing when $Q_{ji}<0$.
\end{lemma}
Proofs: Appendix~S1.

Together the two lemmas give the two-sided monotonicity that ex-post incentive compatibility requires under correlated values~\cite{rtc2016}. Regularity needs only a Gaussian conditional. Cross-monotonicity is where the M-matrix condition enters, through the sign of $Q_{ji}$, and it holds for every M-matrix $Q$, which Theorem~\ref{thm:robust} builds on.

\subsection{Incentive Compatibility and Optimality}

\begin{lemma}[Threshold well-posedness]\label{thm:threshold}
For the provisional winner $i^*$, Equation~\eqref{eq:threshold} has a unique solution $t_0\in(0,s_{i^*}]$ (proof: Appendix~S2).
\end{lemma}

The incentive guarantee is the standard ex-post characterization for correlated private values (\cite{rtc2016}, Prop.~5.1) applied to our allocation rule. A bidder cannot gain by misreporting when winning is monotone in its own report and the price is the lowest report that still wins. Lemmas~\ref{lem:reg} and~\ref{lem:cross} supply the monotonicity and Lemma~\ref{thm:threshold} the price. A winner pays $t_0\le s_{i^*}$ and losers pay zero (full proof: Appendix~S2).

\begin{theorem}[EPIC and EPIR]\label{thm:epic}
The mechanism of Section~\ref{sec:mech} is ex-post incentive compatible (EPIC) and ex-post individually rational (EPIR); with private values ($v_i=s_i$), EPIC is equivalent to dominant-strategy IC~\cite{rtc2016}.
\end{theorem}

\begin{theorem}[Revenue optimality]\label{thm:opt}
If the learned distribution equals the true distribution, the mechanism maximizes expected revenue among all EPIC and EPIR mechanisms with no positive transfers and finite expected payments (two conditions the unbounded Gaussian support requires). The theorem is Theorem~5.3 of~\cite{rtc2016}. This paper verifies its hypotheses for every prior in the family and supplies the revenue identity on the unbounded support (Theorem~S3).
\end{theorem}
\begin{proof}[Condition-by-condition verification]
\cite{rtc2016}, Thm.~5.3 requires (i) correlated private values (Section~\ref{sec:model}), (ii) affiliation (Proposition~\ref{prop:mtp2}), (iii) regularity at every $\bs_{-i}$ (Lemma~\ref{lem:reg}, so ironing never arises), (iv) cross-monotonicity (Lemma~\ref{lem:cross}), and (v) single-item feasibility (rank-one matroid). Under (i)--(v), the revenue-maximal EPIC/EPIR mechanism in that characterization is exactly ``highest nonnegative conditional virtual value wins at the threshold payment''; boundary and tie-breaking cases are Appendix~S3.
\end{proof}

\begin{corollary}[Strict advantage over IPV-Myerson]\label{cor:strict}\sloppy
Suppose the mechanism uses the \emph{true} distribution and the IPV baseline uses its \emph{true} regular marginals. The weak inequality $\E[\mathrm{Rev}_{\textsc{AffMyerson}}]\ge\E[\mathrm{Rev}_{\mathrm{IPV}}]$ is immediate from Theorem~\ref{thm:opt}, since IPV-Myerson is itself EPIC and EPIR; the inequality is \emph{strict} whenever some $Q_{ij}<0$, because the two allocations then disagree on a set of positive probability (proof Appendix~S3).
\end{corollary}

\subsection{Truthfulness under an Estimated Prior}\label{sec:robust}

The results so far use the true prior. We now let the mechanism run on any prior in the family.

\begin{theorem}[Truthfulness for every prior in the family]\label{thm:robust}
Let the mechanism run with \emph{any} M-matrix parameters $(\hat\bmu,\hat Q)$ satisfying Assumption~\ref{as:mvn}, not necessarily equal to the true ones. Then it remains EPIC, hence dominant-strategy IC, and EPIR.
\end{theorem}
\begin{proof}
EPIC and EPIR are \emph{structural} properties of the allocation and payment rules (\cite{rtc2016}, Prop.~5.1): monotone allocation and threshold payments, neither of which references the true distribution. Lemmas~\ref{lem:reg} and~\ref{lem:cross} hold for every M-matrix $\hat Q$, so the argument for Theorem~\ref{thm:epic} goes through with $\hat Q$ in place of $Q$.
\end{proof}

The M-matrix parameterization therefore separates learning the prior from incentives. Cr\'emer and McLean~\cite{cremer1988} extract the full surplus under correlation but need the exact prior. We give up full extraction, and in exchange truthfulness holds under any prior error.

Proposition~\ref{prop:stab} bounds what an error in the prior can cost, measured in the quantities the mechanism uses, each bidder's conditional mean and conditional scale.

\begin{proposition}[Revenue stability]\label{prop:stab}
For any two priors in the family, the revenue difference is at most an explicit constant times the average error in the conditional means and scales, plus a term for the profiles where the two priors disagree about who wins or whether to sell. Revenue is therefore continuous in the fitted prior. Formal statement, constant, and proof: Appendix~S6.
\end{proposition}

\paragraph{The second family: log-normal values.}
The same guarantees hold for a second family, the log-normal, the long-tailed model commonly used in empirical work on advertising values~\cite{lahaie2007,edelman2010,ostrovsky2011}. Under the transformation $y_i=\log s_i$ affiliation is preserved, but the Jacobian puts a factor $s_i$ in front of the virtual value, and regularity is no longer automatic.
\begin{theorem}[The log-normal family: optimality and exact regularity]\label{thm:loglink}
Let $\mathbf y=\log\bs$ follow the M-matrix Gaussian of Assumption~\ref{as:mvn}, with conditional moments $\mu_{i\mid-i}$ and $\sigma_{i\mid-i}=Q_{ii}^{-1/2}$. Then (i) the conditional virtual value is $\psi_i(s_i\mid\bs_{-i})=s_i\bigl[1-\sigma_{i\mid-i}R(z_i)\bigr]$ with $z_i=(\log s_i-\mu_{i\mid-i})/\sigma_{i\mid-i}$; (ii) allocating to the highest nonnegative conditional virtual value at the threshold payment is EPIC and EPIR for every such prior, and when the prior is correct it maximizes expected revenue among EPIC and EPIR mechanisms with no positive transfers and finite expected payments; (iii) $\psi_i$ is increasing on its entire support if and only if $\sigma_{i\mid-i}\le\sigma^\star=\min_{z}\{2/R(z)-z\}=1.5176$
\end{theorem}
The threshold characterizes global regularity. Neither truthfulness nor optimality needs it, because virtual values increase throughout their nonnegative region and negative virtual values never receive the item. It is loose (one standard deviation in log scale multiplies a value by $4.6$, while the largest conditional log-scale in our fitted priors is $0.64$), and the parameterization keeps every output globally regular by a floor of $0.45>1/(\sigma^\star)^2$ on the precision diagonal, a design choice that the guarantees do not need. The experiments use the normal family as their main line.
\section{Context-Dependent Neural Parameterization}\label{sec:neural}
The theory above holds for any M-matrix prior (Figure~\ref{fig:pipeline}). This section is about learning one from data. A network reads the embedding of the conversation and outputs, for each ad opportunity, a mean vector and a precision matrix, and the difficulty is to keep every output inside the M-matrix class. We give two constructions that do so by design: a one-factor form, which is small and matches the common-demand story, and a spectral form, which covers the whole class. We then describe two alternatives that we use as controls and a convex fit that needs no network. A frozen sentence encoder (all-MiniLM-L6-v2~\cite{reimers2019,wang2020}) embeds the query, and shared MLP (multilayer perceptron) heads output, for each context, the means $\mu_i=\mathrm{Softplus}(\cdot)$ and the precision parameters below. Training is by maximum likelihood in precision form on historical bids and never forms the covariance matrix.

\paragraph{Default construction (diagonal plus rank one).}
The economic premise of Section~1 is a \emph{common demand factor}, and a prior with one common factor is automatically an M-matrix prior:
\begin{proposition}[Single factor $\Rightarrow$ M-matrix]\label{prop:rankone}
If $\Sigma=D+\lambda vv^\top$ with $D$ positive diagonal, $\lambda>0$, and $v\ge0$ entrywise, then $Q=\Sigma^{-1}$ is a symmetric positive definite M-matrix.
\end{proposition}
The proof is a Sherman--Morrison computation (Appendix~S5).
Proposition~\ref{prop:rankone} gives a sufficient condition. Every one-factor prior is in the class, but the class is larger, and the spectral construction below covers the rest. Under the one-factor construction the network outputs only $d\in\R^n_{>0}$ and $v\in\R^n_{\ge0}$ ($2n$ parameters), and $v_i$ is brand $i$'s sensitivity to the category demand shock. On real queries this default matches the full class to $0.02$ cents (Section~\ref{sec:exp-realq}).
\paragraph{General construction: spectral M-matrix.}
When correlation beyond a common factor is needed, we want a network head that can output any matrix in the class and no matrix outside it. The head outputs a positive diagonal $D$, a symmetric nonnegative matrix $C_0$ with zero diagonal, and a scalar $g$, and the prior's precision is $Q=D^{1/2}(I-\alpha C_0)D^{1/2}$ with $\alpha=\mathrm{sigmoid}(g)/\lambda_{\max}(C_0)$. Dividing by the largest eigenvalue of $C_0$ keeps $I-\alpha C_0$ positive definite, multiplying by $D^{1/2}$ on both sides preserves positive definiteness, and the off-diagonal entries of $Q$ are nonpositive because those of $C_0$ are nonnegative. Every output is therefore an M-matrix, and every M-matrix can be written this way, so the head covers the whole family except for two numerical-stability margins, a small gap below $\alpha$'s upper limit and a floor on $D$ (Appendix~S5). The main-line experiments use this construction.

\paragraph{Two alternative constructions.}
Strict diagonal dominance also yields M-matrices but caps the correlation the prior can express, and at $n=4$ the cap binds (Section~\ref{sec:exp-fit}). Unconstrained learning leaves $24$--$27\%$ of the off-diagonal entries outside the class and would need a repair step at every auction (Appendix~S5). Both stay truthful by Theorem~\ref{thm:robust}, so only revenue is at stake. The affiliation-constrained Gaussian maximum-likelihood fit needs no network at all~\cite{luz2019,slawski2015}, and Section~\ref{sec:exp-realq} shows it matches the network on real queries. The network extends that convex fit to conversation-dependent priors.
\section{Experiments}\label{sec:exp}

We evaluate \textsc{AffMyerson}, run with the normal prior family, on a semi-synthetic LLM-advertising benchmark where queries and bids are produced by LLMs, a hidden demand shock the platform cannot see induces correlation, and all revenue is settled on held-out LLM-agent bids. The prior is therefore estimated, never exact, which is the regime Theorem~\ref{thm:robust} is designed for. Section~\ref{sec:exp-data} describes the benchmark, Sections~\ref{sec:exp-fit} and~\ref{sec:exp-rev} the fitted priors and the revenue comparison, and Sections~\ref{sec:exp-rho}--\ref{sec:exp-perturb} test where the rent comes from and when it fails.

\subsection{Benchmark}\label{sec:exp-data}

\textbf{Queries.} We use \texttt{gpt-4o-mini} to generate $600$ designed conversational queries in the outdoor category over a ``subcategory $\times$ purchase intent'' grid with four intent tiers (transactional, comparison, informational, adjacent), plus $180$ special queries (niche and brand-specific) for the $n{=}4$ study.
\textbf{Advertisers.} Four price-stratified brands (flagship, premium, value, clearance) each bid independently per (query, demand level) through an LLM agent call ($62{,}400$ calls; prompts and brand profiles Appendix~S7). Because each brand bids without seeing the others, any subset of brands is a valid smaller market. The main line is $n{=}2$ (premium$+$value, $12{,}000$ auctions).
\textbf{Hidden demand shock.} A topic-independent demand signal $u$ enters only the agents' prompts and not the platform's context encoding. All brands in a scenario read the same $u$, which creates the correlation the mechanism prices. Within-query correlation is $+0.66$, and all pairwise within-query correlations are positive (Appendix~S7).
\textbf{Protocol.} Queries are split $6{:}1{:}3$ into train/validation/test, so no test conversation is seen in training. Learned quantities are means over $5$ training seeds (seed deviations in the tables), and the bidders are \texttt{gpt-4o-mini}~\cite{gpt4o}. Full prompts, brand profiles, and hyperparameters are in Appendix~S7, and the code will be made publicly available.

\subsection{Fitting the Conditional Distribution}\label{sec:exp-fit}

We first ask how well the parameterization fits. Of four priors fitted on the full $n{=}4$ data (table Appendix~S5), strict diagonal dominance learns an informative context mean ($R^2$ $0.60$ on unseen queries) but allows little correlation (mean fitted correlation $\bar\rho$ $0.22$), and an unconstrained precision fit learns neither. The deployed spectral fit reads the context mean poorly ($R^2\approx0$) but captures the correlation ($\bar\rho$ $0.47$), and revenue depends on that channel (Section~\ref{sec:exp-realq}). Per-cell likelihoods favor the log-normal overall but the normal in the high-value region where revenue concentrates (Appendix~S10).

\subsection{Revenue on Held-Out LLM-Agent Bids}\label{sec:exp-rev}

\begin{table}[t]\centering\small
\caption{Revenue on held-out LLM-agent bids ($n{=}2$, cents, mean$\pm$standard deviation over $5$ training seeds where a prior is learned; percentages use unrounded means). The last column is \textsc{AffMyerson}'s lead over that row. Five-seed and five-split reruns: Appendix~S14.}\label{tab:main}
\begin{tabular}{@{}lr r@{}}\toprule
Mechanism & Revenue & Lead \\\midrule
Second price (SPA) & $28.39$ & $+33.3\%$ \\
Contextual-reserve SPA & $32.98\pm0.18$ & $+14.7\%$ \\
IPV-Myerson (same prior) & $34.66\pm0.19$ & $+9.2\%$ \\
Tuned-reserve SPA ($r^*{=}32$) & $35.62$ & $+6.2\%$ \\
Ronen lookahead~\cite{ronen2001} & $35.67\pm0.08$ & $+6.1\%$ \\
CAN, rank scores, no context~\cite{huo2025can} & $37.43\pm0.07$ & $+1.1\%$ \\
\textsc{AffMyerson} & $\mathbf{37.84\pm0.06}$ & \\\bottomrule
\end{tabular}
\end{table}

\textbf{Baselines.} Besides the second-price auction (SPA, here equal to VCG) and SPA with the reserve tuned on the training set ($r^*{=}32$), we run a contextual-reserve SPA (ctxSPA), which sets per-advertiser reserves from the same learned marginals; IPV-Myerson, which uses \textsc{AffMyerson}'s exact means and marginal deviations but does not condition; Ronen's lookahead auction~\cite{ronen2001} with the same learned prior (LA), which gives the item to the highest bidder at the larger of the second bid and its conditional monopoly price; and CAN~\cite{huo2025can}, a neural mechanism whose rank-score networks enforce the single-crossing conditions, trained on the same $7{,}200$ training bids with no prior and no context (implementation Appendix~S7).

In Table~\ref{tab:main}, \textsc{AffMyerson} earns $37.84$ cents. The strongest baseline is CAN at $37.43$, and the strongest closed-form baselines are LA at $35.67$ and the tuned-reserve SPA at $35.62$. The affiliation rent over IPV-Myerson is $\mathbf{+9.2\%}$ and the margin over the tuned reserve is $\mathbf{+6.2\%}$ (seed deviations in the table, reruns Appendix~S14).

Two comparisons show where the margin comes from. The ctxSPA comparison shows that the margin does not come from contextual reserves, since reserves set on the noisy context model earn $7.4\%$ less than the tuned constant. LA shows that conditional pricing alone is not enough. Pricing the highest bidder conditionally earns $+2.9\%$ over IPV but only ties the tuned reserve. \textsc{AffMyerson} also ranks bidders by conditional virtual value (its winner differs from the highest bidder in $5.6\%$ of auctions) and earns $+6.1\%$ over LA, so the margin comes from conditional allocation. CAN, at $37.43$, recovers most of the rent from the bids alone. \textsc{AffMyerson} earns more, $+1.1\%$, and every \textsc{AffMyerson} seed exceeds every CAN seed. Beyond revenue, \textsc{AffMyerson} is a closed-form rule with a prior the platform can inspect, priced in $O(n^2)$ time, and its incentive compatibility and the robustness guarantee of Theorem~\ref{thm:robust} hold by construction.

Revenue does not come at the cost of efficiency, since against IPV-Myerson \textsc{AffMyerson} raises both welfare, the winner's value ($53.8$ vs.\ $48.2$ cents), and the trade rate ($89.8\%$ vs.\ $69.7\%$).

\textbf{What the conversation is worth.} In a market where the conversation matters to one advertiser and not to the other, \textsc{AffMyerson} earns $24.8\%$ more than CAN, and the gap is the value of the information that bids alone do not carry. Markets like this exist. A marketplace that sells everything bids about the same on any outdoor query, while a specialist brand's value depends on what the conversation says. We write $400$ conversations in $200$ pairs, where the same product words carry either a high-value cue (an upcoming trip, a broken item) or a low-value one (a school project, an item already bought). The marketplace's value is $30$ cents in every auction, and the specialist's is $200$ on a high cue and $100$ on a low one, each with a spread of $25$. Values are drawn from the Gaussian family, so the prior is fitted as on the main line. The two advertisers' values are independent here, so the whole gain comes from the part of our design that reads the conversation, the context-dependent prior. On $2{,}400$ held-out auctions \textsc{AffMyerson} earns $104.8$ cents and CAN $84.1$. CAN earns as much as the bids allow. Without the context, the most any mechanism can earn is what Myerson's auction on the marginal bid distribution earns, $84.4$, and CAN reaches it. The rest is earned on the low-cue conversations, where a mechanism that sees only bids sells to the marketplace for $24.7$ cents while the fitted prior places the specialist's value near $100$ and sells to it for $66.8$. On the main line every advertiser reads the same conversation, the bids move together, and CAN reads the conversation through the bids. Here the marketplace's bid does not reveal the cue, so only a prior that reads the conversation can price it, at no cost in truthfulness (Theorem~\ref{thm:robust}; design and full table Appendix~S7).
\subsection{Rent versus Affiliation Strength}\label{sec:exp-rho}

The rent rises and falls with the correlation in the data (Figure~\ref{fig:rho}). To check that the rent comes from the correlation, we remove the correlation without changing the marginals. Within each query we re-pair the second brand's bids across scenarios, so that full shuffling zeroes the within-query correlation, partial shuffling interpolates, and each brand's marginal is unchanged element by element. The prior is refit and the reserve re-tuned at every level. Figure~\ref{fig:rho} shows the rent over IPV-Myerson and the margin over the tuned-reserve SPA along this curve. We observe that the rent rises monotonically with $\rho$, is statistically zero at $\rho\approx0$, and tracks the fitted $\bar\rho$, while the margin over the tuned reserve becomes significant only above $\rho\approx0.3$. These thresholds are specific to this benchmark, so a platform should measure its own before it adopts conditioning.

\begin{figure}[t]\centering
\includegraphics[width=0.58\linewidth]{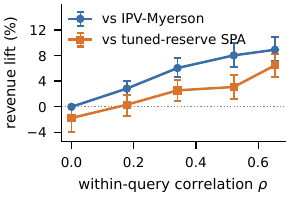}
\Description{Revenue lift over IPV-Myerson and over the tuned-reserve second-price auction as the within-query correlation is varied.}
\caption{Shuffling scan (prior refit and reserve re-tuned at every level): lift over IPV-Myerson (blue) and over the tuned-reserve SPA (orange) against within-query $\rho$. The panel that tracks the fitted $\bar\rho$ is Appendix~S9.}\label{fig:rho}
\end{figure}

\subsection{$n=4$ Stratified Markets}\label{sec:exp-n4}

We next ask whether the margin survives more bidders. On the full four-brand data the revenue is $49.05$ cents, $\mathbf{+4.2\%}$ over the tuned-reserve SPA and $\mathbf{+6.2\%}$ over matched IPV. Across the $11$ brand submarkets (table Appendix~S14) we observe that the fitted $\bar\rho$ alone does not rank the markets. The three lowest-rent submarkets each pair a premium brand with a clearance brand whose near-constant bids carry no usable signal, so the rent depends on how informative competitors' reports are as well as on correlation. The median rent across submarkets is flat in $n$, so added competitors bring conditioning information that offsets dilution.

\subsection{Real Search Queries}\label{sec:exp-realq}

To test whether the findings depend on simulated queries, we build a real-query benchmark of $580$ genuine user queries from Amazon ESCI~\cite{esci2022} ($274$) and MS MARCO~\cite{msmarco2016} ($306$), filtered by embedding ranking and two LLM gates and passed unchanged through the same bidding pipeline (Appendix~S16).

\textbf{The rent persists on real queries.} The premium brand bids $4.60$ times more on transactional than on informational queries and the two brands' bids correlate at $0.722$ within a query, so both market features of Section~\ref{sec:exp-ablation} reappear on queries we did not write. A short search query gives the prior little to read (Appendix~S16). Even so, the rent is $\mathbf{+4.5\%}$ and the margin over the tuned reserve is $\mathbf{+4.7\%}$ (per-source and channel diagnostics: Appendix~S16).

\textbf{What the network adds.} On short search queries, where the text gives the prior little to read, a convex $\mathrm{MTP}_2$-constrained Gaussian fit with no network~\cite{luz2019,slawski2015} earns the same revenue as the embedding network ($43.65$ vs.\ $43.48$ cents, within seed noise) and a one-factor prior (Proposition~\ref{prop:rankone}) is within $0.02$ of it, so the rent does not depend on the fitting method. On the designed conversational queries, where the text does predict values, an unconstrained network repaired into the class earns $2.0\%$ more than the best convex fit ($49.29$ vs.\ $48.33$ cents; Appendix~S11).

\subsection{How Bidding Agents Turn Common Demand into Bids}\label{sec:exp-ablation}

We now ask where the correlation the mechanism prices comes from. Our bidders are LLM agents acting as advertisers. Each receives a brand profile of behavioral traits, the conversation, and the shared demand signal, and returns a bid. We check their behavior against two features of real sponsored-search markets, that bids rise with commercial intent~\cite{jansen2011} and that rivals' values move with common demand shocks~\cite{milgrom1982}, and the agents reproduce both. On the main line the premium brand bids $5.76$ times more on transactional than on informational queries, and the two brands' bids within a query correlate at $+0.66$.

Within a query, the shared demand shock is the only input the agents have in common (Appendix~S7), so the question is how differently prompted agents translate it into bids. With queries and shocks fixed, within-query correlation is $+0.78$ when the agents judge the conversation and $+0.06$ when they are prompted to calculate (seven schemes, two models; table Appendix~S7), and the calculating agents also concentrate their bids near zero. This variation matters for auction design. On the full-scale calculation-type bids \textsc{AffMyerson} still earns $+2.5\%$ over IPV-Myerson but less than the tuned reserve, so what conditioning is worth depends on how agents express demand in their bids.

\subsection{Robustness to Prior Perturbation}\label{sec:exp-perturb}

Last, we corrupt the learned prior deliberately, in twelve controlled ways ($n{=}4$, seed $0$), and find that revenue falls gradually while truthfulness is unaffected. We observe three things. Mean noise $\varepsilon$ reduces the rent monotonically, from $+6.4\%$ at $\varepsilon{=}0$ to $-4.3\%$ at $\varepsilon{=}20$, consistent with Proposition~\ref{prop:stab}, which bounds the loss from above only. Setting the correlation to zero gives the revenue of IPV-Myerson. No corruption produces an incentive violation ($0$ of $28{,}800$ probes), as Theorem~\ref{thm:robust} requires. At $\varepsilon{=}20$ \textsc{AffMyerson} falls below the tuned-reserve SPA, so a platform with an unusable prior should prefer prior-free second pricing.
\section{Scope and Limitations}\label{sec:conclusion}

Our LLM-agent benchmarks show strong positive within-query correlation, consistent with affiliation, and IPV mechanisms give up the rent on such data. Five limitations remain. (i) The guarantee covers nonnegative dependence only, so negatively correlated substitutes are outside it. (ii) The values are correlated private values and not interdependent values. (iii) The auction sells one item, and multi-slot allocation needs a combinatorial extension. (iv) Bids are LLM-agent outputs, since no public dataset offers jointly observed real bids on one opportunity. (v) Learned priors invite history shading~\cite{kanoria2021}, and the channel that earns the rent also transmits a rival's misreport. Averaging the prior across periods and capping the conditional-mean shift limit both, at little or no revenue cost (Appendices~S12 and~S13).

\clearpage
\appendix
\renewcommand{\thesection}{S\arabic{section}}
\setcounter{section}{0}
\noindent{\Large\bfseries Supplementary Material}\par\medskip
\noindent Sections S1--S16 below are the supplementary material. ``Main text'' refers to Sections 1--8.\par\bigskip

\section{Complete Gaussian derivations}

\subsection{Affiliation and the M-matrix sign condition (main text, Proposition~3.1)}
The log-density of $\N(\bmu,Q^{-1})$ is $\ln f(\bs)=\mathrm{const}-\tfrac12(\bs-\bmu)^\top Q(\bs-\bmu)$, whose cross-partials $\partial^2\ln f/\partial s_i\partial s_j=-Q_{ij}$ are constant in $\bs$. $\mathrm{MTP}_2$ requires them to be nonnegative. If some $Q_{ij}>0$ the requirement fails; if all $Q_{ij}\le 0$ then $\ln f$ is supermodular on $\R^n$, which implies $\mathrm{MTP}_2$ by Karlin and Rinott (1980, \emph{J.\ Multivariate Anal.}\ 10, Theorem~3.1).

\subsection{Conditioning from the precision matrix}
Partition the quadratic form for bidder $i$:
\[
(\bs-\bmu)^\top Q(\bs-\bmu)
=Q_{ii}(s_i-\mu_i)^2
+2(s_i-\mu_i)\sum_{j\ne i}Q_{ij}(s_j-\mu_j)+C(\bs_{-i}).
\]
Completing the square gives
\[
s_i\mid\bs_{-i}\sim
\N\!\left(\mu_i-\frac{1}{Q_{ii}}\sum_{j\ne i}Q_{ij}(s_j-\mu_j),
\frac{1}{Q_{ii}}\right).
\]
Only row $i$ of $Q$ is needed.

\subsection{Regularity and cross-monotonicity}
Let $R(z)=(1-\Phi(z))/\phi(z)$. Direct differentiation gives
$R'(z)=zR(z)-1$. Hence, for
$\psi_i(s_i\mid\bs_{-i})=s_i-\sigma_iR(z_i)$,
\[
\frac{\partial\psi_i}{\partial s_i}=2-z_iR(z_i)>1.
\]
For $z_i>0$, the standard Mills bound gives $R(z_i)<1/z_i$; for
$z_i\le0$, $z_iR(z_i)\le0$. For $j\ne i$,
\[
\frac{\partial\mu_{i\mid-i}}{\partial s_j}=-\frac{Q_{ij}}{Q_{ii}}\ge0.
\]
Because $R$ is strictly decreasing, raising $s_j$ weakly lowers $\psi_i$,
strictly when $Q_{ij}<0$.

\paragraph{Mass below zero of the fitted priors (main text, Remark~1).} A prior whose marginal mean satisfies $\mu_i(\mathbf{e})>3\sqrt{(Q(\mathbf{e})^{-1})_{ii}}$ puts negligible mass ($<0.14\%$) below zero. The fitted spectral priors do not satisfy this on most contexts: on the seed-$0$ prior and test set the ratio $\mu_i/\sqrt{\Sigma_{ii}}$ has minimum $0.37$ and exceeds $3$ on only $14\%$ of context--bidder pairs, so on many contexts the prior puts non-negligible mass below zero (\texttt{kappa\_box.py}). Nothing formal depends on it, and in the replay every signal is an observed bid in $[2,185]$ cents.

\paragraph{Parameters of the main text's Figure~2 and the example of Section~4.3.} The illustrative two-bidder prior has $\bmu=(45,35)$ cents, marginal deviations $\sigma=(13,12)$ cents and $\rho=0.7$, i.e.\ $\Sigma=\bigl(\begin{smallmatrix}169&109.2\\109.2&144\end{smallmatrix}\bigr)$ and precision $Q=\Sigma^{-1}=\bigl(\begin{smallmatrix}0.01160&-0.00880\\-0.00880&0.01362\end{smallmatrix}\bigr)$, an M-matrix. Conditional scales $\sigma_{i\mid-i}=Q_{ii}^{-1/2}=(9.28,\,8.57)$ against marginal $(13,12)$; conditional-mean slopes $\beta_{12}=-Q_{12}/Q_{11}=0.758$ and $\beta_{21}=0.646$, so a bid of $61$ by bidder 2 moves bidder 1's conditional mean to $45+0.758\cdot26=64.7$. IPV-Myerson reserves (the dashed lines of the left panel) solve $\varphi^{\mathrm{IPV}}_i(r)=0$: $r_1=34.8$, $r_2=27.3$. Expected revenue over the prior, by quadrature on an $801^2$ grid over $[-60,200]^2$ with $60$-step bisection payments: \textsc{AffMyerson} $36.89$, IPV-Myerson $35.00$, difference $+1.90$ cents ($+5.4\%$). Script \texttt{draw\_fig2.py} (with \texttt{fig2\_core.py}) regenerates the figure and prints these numbers.

\section{Incentives and payments}

\paragraph{Numerical tolerance of the payment (main text, Section~4.2).} Bisection with $60$ steps on the bracket $[0,s_{i^*}]$ locates $t_0$ to $2^{-60}s_{i^*}<2\times10^{-16}$ cents for bids capped at $200$. The mechanism is exactly EPIC; only the payment is evaluated numerically, to this tolerance plus the floating-point error of the Mills ratio (same order).

\begin{stheorem}[EPIC and EPIR]
For any mean vector and any SPD M-matrix used by the mechanism, allocation to
the highest nonnegative conditional virtual value with critical payment is
ex-post incentive compatible and ex-post individually rational.
\end{stheorem}
\begin{proof}
Fix $\bs_{-i}$. Raising $s_i$ strictly raises bidder $i$'s score and weakly
lowers every competing score, so the binary allocation $x_i(s_i,\bs_{-i})$ is
monotone. A monotone single-parameter allocation with its critical payment is
dominant-strategy incentive compatible. If bidder $i$ wins truthfully, its
report is at least the critical threshold, so payment is at most $s_i$.
Losers pay zero. Thus utility is nonnegative for every report profile.
\end{proof}

\begin{stheorem}[Threshold existence and uniqueness]\label{supp:thm:threshold}
For a provisional winner $i$ with report $s_i$, the score gap in the payment
equation has a unique zero $t_0\in(0,s_i]$.
\end{stheorem}
\begin{proof}
Let
\[
D(t)=\psi_i(t\mid\bs_{-i})
-\max\{0,\max_{j\ne i}\psi_j(s_j\mid t,\bs_{-\{i,j\}})\}.
\]
The first term is strictly increasing in $t$ and every term inside the maximum
is nonincreasing, so $D$ is strictly increasing and continuous. Since
$\psi_i(0\mid\bs_{-i})<0$, $D(0)<0$. Since $i$ is the selected provisional
winner, $D(s_i)\ge0$. The intermediate value theorem therefore gives a zero
in $(0,s_i]$, and strict monotonicity gives uniqueness. Hence the bracket
$[0,s_i]$ contains exactly one root; the implementation bisects it for $60$
steps (\texttt{run\_auctions4.py}), reaching resolution $2^{-60}s_i<2\times10^{-16}$
cents for bids capped at $200$.
\end{proof}

\paragraph{Payments need not increase in every competitor's report when $n\ge3$.}
Cross-monotonicity concerns the \emph{allocation}. The winner's threshold $t_0$ solves $\psi_w(t_0\mid\bs_{-w})=\max\{0,\max_{j\ne w}\psi_j(s_j\mid t_0,\bs_{-\{w,j\}})\}$; a third bidder's report enters both sides. For $n=2$ both channels push $t_0$ up. For $n\ge3$ they can conflict: with $Q=\bigl(\begin{smallmatrix}1&0&0\\0&1&-0.8\\0&-0.8&1\end{smallmatrix}\bigr)$ (eigenvalues $0.2,1,1.8$, an SPD M-matrix), $\bmu=(5,5,0)$, $s_1=20$, $s_2=10$, raising $s_3$ from $0$ to $1$ leaves bidder $1$ the winner but lowers its payment from $10.000$ to $9.968$: bidder $1$ is uncoupled from $3$, while $3$'s higher report lowers $\psi_2$ and with it the competing threshold. The main text's ``$\partial t_0/\partial s_j>0$'' is therefore stated for the bivariate case only.

\section{Revenue identity on unbounded Gaussian support}

\begin{stheorem}[Gaussian-tail extension]
When the learned prior equals the true Gaussian prior, AffMyerson maximizes
expected revenue among EPIC, EPIR, no-positive-transfer mechanisms with finite
expected payments.
\end{stheorem}
\begin{proof}
Fix $\bs_{-i}$ and abbreviate the conditional CDF and density by $F$ and $f$.
For a monotone allocation $x(t)$, the envelope theorem gives
\[
u(t)=u(a)+\int_a^t x(z)\,dz.
\]
EPIR and no positive transfers imply $x(t)=0$ for $t<0$ and $u(0)=0$.
Therefore, for $t\ge0$,
\[
p(t)=t x(t)-\int_{0}^{t}x(z)\,dz,
\]
while $p(t)=x(t)=0$ for $t<0$. Since the Gaussian has a finite first moment,
Tonelli's theorem applies:
\begin{align*}
\E[p\mid\bs_{-i}]
&=\int_{0}^{\infty}t x(t)f(t)\,dt
-\int_{0}^{\infty}f(t)\int_{0}^{t}x(z)\,dz\,dt\\
&=\int_{0}^{\infty}x(t)\{t f(t)-(1-F(t))\}\,dt\\
&=\E\!\left[x(s_i,\bs_{-i})
\left(s_i-\frac{1-F(s_i)}{f(s_i)}\right)\middle|\bs_{-i}\right].
\end{align*}
Summing over bidders converts revenue to conditional virtual surplus.
Pointwise, single-item feasibility is maximized by allocating to the highest
nonnegative conditional virtual value. The preceding monotonicity proof makes
this allocation implementable, and critical payments attain the identity.
Finally $\psi_i<s_i$, so any allocated type is positive and EPIR holds.
\end{proof}

\paragraph{Boundary and ties (main text, Theorem~5.5).} Two events need a convention: a tie for the highest conditional virtual value, and a maximum exactly equal to zero. Under any nondegenerate Gaussian prior both have probability zero, so every fixed convention yields the same expected revenue and the same incentive properties (the allocation remains monotone and the payment remains the threshold). The implementation allocates when the maximum is $\ge0$ and breaks exact ties by the lowest bidder index (\texttt{argmax}); the threshold bracket $[0,s_{i^*}]$ of the main text's Lemma~5.3 is unaffected.

\paragraph{Dominance over IPV: full proof of the main text's Corollary~5.6.}
The main text states the corollary without a sketch; the argument is as follows. IPV-Myerson is a feasible EPIC/EPIR mechanism whose revenue equals its expected \emph{true} conditional virtual surplus (Roughgarden and Talgam-Cohen 2016, Prop.~5.2), while \textsc{AffMyerson} maximizes that integrand pointwise, which gives the weak inequality. For strictness: when $Q_{ij}<0$ the moving threshold of \textsc{AffMyerson} differs from the fixed IPV reserve for all but one value of $s_j$, and on the open interval between them exactly one mechanism allocates to $i$ (the one with the larger integrand) on a set of positive MVN measure.
Regular IPV-Myerson is itself EPIC and EPIR, hence the theorem gives weak
revenue dominance under the true prior. The inequality is strict exactly when
the IPV allocation fails to maximize true conditional virtual surplus on a
positive-probability set; when $Q_{ij}<0$ for some pair, such a set can be
exhibited constructively, so strictness there is derived, not assumed.

\begin{proof}[Constructive proof of strictness]
IPV-Myerson is a feasible EPIC/EPIR mechanism: its allocation is monotone in
each own signal (regular marginals) with threshold payments. By the revenue
identity above (equivalently, Prop.~5.2 of Roughgarden and Talgam-Cohen),
\emph{any} EPIC/EPIR mechanism's expected revenue equals its expected
conditional virtual surplus, $\E_{\bs}[\sum_i x_i(\bs)\psi_i(s_i\mid\bs_{-i})]$
with the \emph{true} $\psi_i$ common to both mechanisms. AffMyerson maximizes
the integrand pointwise, so its revenue weakly dominates. Strictness: pick
$i,j$ with $Q_{ij}<0$ and fix the other coordinates. Let $r_i$ solve
$\varphi^{\mathrm{IPV}}_i(r_i)=0$ (independent of $\bs_{-i}$) and
$\rho_i(\bs_{-i})$ solve $\psi_i(\rho_i\mid\bs_{-i})=0$ (unique by the
unconditional-regularity lemma). By cross-monotonicity,
$s_j\mapsto\psi_i(r_i\mid\bs_{-i})$ is strictly decreasing, so
$\rho_i(\bs_{-i})\neq r_i$ except at most one $s_j$. Choose $s_j$ low enough
that competitor virtual values are negative on the relevant compact set; then
for $s_i$ in the nonempty open interval between $r_i$ and $\rho_i$, exactly
one mechanism allocates to $i$: if IPV allocates where AffMyerson does not
($\rho_i>r_i$), the IPV integrand is $\psi_i<0$ versus $0$; in the opposite
case AffMyerson's integrand is $\psi_i>0$ versus $0$. Either way IPV's
integrand is strictly smaller on an open set of positive MVN measure;
integrating gives the strict inequality.
\end{proof}

\section{Log-link regularity: the complete characterization}

Under the identity link (main text), $\partial\psi_i/\partial s_i=2-z_iR(z_i)>1$
by the standard bound $zR(z)<1$, so regularity is automatic. This section gives
the complete characterization promised in the main text for the multiplicative
(log-normal) specification.

\begin{slemma}[Conditional virtual value under the log link]
Let $y_i=\log s_i$ follow the M-matrix MVN of the main text's assumption. Then
$s_i\mid\bs_{-i}$ is log-normal and
\[
\psi_i(s_i\mid\bs_{-i})=s_i\bigl[1-\sigma_{i\mid-i}R(z_i)\bigr],\quad
z_i=\frac{\log s_i-\mu_{i\mid-i}}{\sigma_{i\mid-i}},\quad
\sigma_{i\mid-i}=Q_{ii}^{-1/2}.
\]
\end{slemma}
\begin{proof}
The conditional law of $y_i$ is $\N(\mu_{i\mid-i},\sigma^2_{i\mid-i})$. Under
$s_i=e^{y_i}$, $F_{i\mid-i}(s)=\Phi(z)$ and
$f_{i\mid-i}(s)=\phi(z)/(s\,\sigma_{i\mid-i})$ (the $1/s$ is the Jacobian), so
$(1-F)/f=s\,\sigma_{i\mid-i}R(z)$, and the general conditional virtual value
$s-(1-F)/f$ gives the display.
\end{proof}

\begin{stheorem}[Regularity threshold $\sigma^\star$]
In the setting above,
\[
\frac{\partial\psi_i}{\partial s_i}=2-\bigl(z_i+\sigma_{i\mid-i}\bigr)R(z_i),
\]
so $\psi_i$ is strictly increasing in $s_i$ iff
$\sup_{z}(z+\sigma_{i\mid-i})R(z)\le2$. Writing $g_\sigma(z)=(z+\sigma)R(z)$, the
map $\sigma\mapsto\sup_z g_\sigma(z)$ is strictly increasing and crosses $2$ at
a unique
\[
\sigma^\star=1.517619\ldots,
\]
so regularity holds iff $\sigma_{i\mid-i}\le\sigma^\star$, equivalently
$Q_{ii}\ge1/(\sigma^\star)^2=0.434184\ldots$ For $\sigma_{i\mid-i}<\sigma^\star$ the
slope is bounded away from zero; at $\sigma_{i\mid-i}=\sigma^\star$ it vanishes at
exactly one point ($z^\star=-0.551$) and $\psi_i$ is still strictly increasing;
for $\sigma_{i\mid-i}>\sigma^\star$ it is negative on an open interval and
regularity fails.
\end{stheorem}
\begin{proof}
Differentiating, with $\partial z_i/\partial s_i=1/(s_i\sigma_{i\mid-i})$ and
$R'(z)=zR(z)-1$ (from $\phi'=-z\phi$),
\[
\frac{\partial\psi_i}{\partial s_i}
=\bigl(1-\sigma_{i\mid-i}R(z_i)\bigr)-R'(z_i)
=2-(z_i+\sigma_{i\mid-i})R(z_i).
\]
Finiteness of the supremum: as $z\to+\infty$, $R(z)=z^{-1}-z^{-3}+O(z^{-5})$
gives $g_\sigma(z)\to1$; as $z\to-\infty$, $R(z)\sim\sqrt{2\pi}e^{z^2/2}$ while
$(z+\sigma)\to-\infty$, so $g_\sigma\to-\infty$; hence the supremum is attained
at an interior point. Monotonicity in $\sigma$: $\partial g_\sigma/\partial
\sigma=R(z)>0$ pointwise, so the supremum is strictly increasing in $\sigma$;
at $\sigma=0$, $\sup_z zR(z)=1<2$, and the supremum diverges as
$\sigma\to\infty$, so continuity gives a unique crossing $\sigma^\star$. The
stationarity condition $R(z)[1+z(z+\sigma)]=z+\sigma$ is transcendental, so
$\sigma^\star$ is determined numerically: $\sigma^\star=1.517619\ldots$
\end{proof}

\paragraph{Constructive guarantee (design choice).}
In the spectral parameterization $Q_{ii}=D_{ii}$ (the coupling matrix has zero
diagonal), so raising the diagonal floor of $D$ to $c_0\ge0.45>0.434184$ makes
$\sigma_{i\mid-i}<\sigma^\star$ hold for \emph{every} network output. Under
the log link, global regularity becomes an architectural property exactly as
affiliation is, with no inference-time checks; by the proposition above this is a design choice, since neither truthfulness nor optimality depends on it. The identity-link deployment in
the main text needs no floor at all.

\paragraph{What $\sigma^\star$ does and does not govern.}
The theorem characterizes textbook regularity, monotonicity of $\psi_i$ on the whole support. Neither truthfulness nor optimality needs it in the single-item setting, as the next proposition (main text, Theorem~5.9(ii)) shows.

\begin{sproposition}[Truthfulness and optimality for every conditional log-scale]
Under the log link with any SPD M-matrix $Q$, the allocation ``highest nonnegative conditional virtual value wins'' with threshold payments is EPIC and EPIR. If the prior is correct, it maximizes expected revenue among all EPIC and EPIR mechanisms with no positive transfers and finite expected payments, for every value of $\sigma_{i\mid-i}$.
\end{sproposition}
\begin{proof}
\emph{Step 1 (the nonnegative region is an up-set).} $\psi_i/s_i=1-\sigma_{i\mid-i}R(z_i)$ is strictly increasing in $s_i$, since $z_i$ increases in $s_i$ and $R$ is decreasing; hence $\{\psi_i\ge0\}=[s_i^\star,\infty)$ for some $s_i^\star\in(0,\infty)$.
\emph{Step 2 (monotone on that region).} On $\{\psi_i\ge0\}$ one has $\sigma R(z_i)\le1$, hence
\[\frac{\partial\psi_i}{\partial s_i}=2-z_iR(z_i)-\sigma R(z_i)\ \ge\ 1-z_iR(z_i)\ >\ 0\]
for every $\sigma>0$,
so $\psi_i$ is strictly increasing wherever it is nonnegative. Cross-monotonicity carries over under the monotone link ($\partial\mu_{j\mid-j}/\partial y_i=-Q_{ji}/Q_{jj}\ge0$ and $\psi_j$ is decreasing in $\mu_{j\mid-j}$). Therefore, as $s_i$ rises, bidder $i$ moves from ``never wins'' ($\psi_i<0$) to a region where its own virtual value rises and every rival's falls: the allocation is monotone in $s_i$ for every $\bs_{-i}$, the threshold $t_0$ is unique, and EPIC/EPIR follow exactly as in the main text's Theorem~5.4.
\emph{Step 3 (revenue identity without regularity).} Fix $\bs_{-i}$ and let $x_i(\cdot\mid\bs_{-i})$ be any allocation rule that is nondecreasing in $s_i$, with the EPIR-binding payment $p_i(s_i)=x_i(s_i)s_i-\int_0^{s_i}x_i(t)\,dt$. Integrating by parts against the conditional density $f_{i\mid-i}$ gives $\E[p_i\mid\bs_{-i}]=\E[x_i\,\psi_i\mid\bs_{-i}]$ provided the boundary term $x_i(s)\,s\,(1-F_{i\mid-i}(s))\to0$ as $s\to\infty$ and $\E[s_i\mid\bs_{-i}]<\infty$. Both hold for the log-normal: $\E[s_i\mid\bs_{-i}]=e^{\mu+\sigma^2/2}<\infty$ and $s(1-\Phi(z))\le s\,\phi(z)/z\to0$. The identity uses only monotonicity and the IR normalization, not regularity. Summing over $i$ and taking expectations, every EPIC/EPIR mechanism with no positive transfers and finite expected payments has $\mathrm{Rev}=\E[\sum_i x_i\psi_i]$.
\emph{Step 4 (pointwise bound and attainment).} Since $\sum_i x_i\le1$ and $x_i\ge0$, $\sum_ix_i\psi_i\le\max\{0,\psi_1,\dots,\psi_n\}$ pointwise, so $\mathrm{Rev}\le\E[\max\{0,\max_i\psi_i\}]$. The mechanism of Step~2 attains this bound pointwise, is EPIC/EPIR by Step~2, and has finite expected payments (payments are bounded by the winner's value, which has finite mean). Hence it is optimal.
\end{proof}

Global regularity therefore matters only for the un-ironed rule to coincide with Myerson's ironed rule on the whole support; where it fails, on part of $\{\psi_i<0\}$, the item is never sold, so nothing changes. We keep $\sigma\le\sigma^\star$ as the condition under which the virtual value is a textbook regular one, and enforce it in the deployed parameterization as a design choice (floor $c_0=0.45$), not as a requirement of the guarantees.

\section{Learned hypothesis class}

\paragraph{Proof of the main text's Proposition~6.1 (single factor $\Rightarrow$ M-matrix).} Let $\Sigma=D+\lambda vv^\top$ with $D$ positive diagonal, $\lambda>0$ and $v\ge0$ entrywise. By the Sherman--Morrison formula, with $c=1+\lambda v^\top D^{-1}v>0$,
\[Q=\Sigma^{-1}=D^{-1}-\frac{\lambda}{c}\,(D^{-1}v)(D^{-1}v)^\top,\]
\[Q_{ij}=-\frac{\lambda\,v_iv_j}{c\,d_id_j}\le0\quad(i\neq j),\]
so the off-diagonal entries are nonpositive, and $\Sigma\succ0$ gives $Q\succ0$. Hence $Q$ is a symmetric positive definite M-matrix.

For every context, the deployed network emits a positive diagonal $D$, a
symmetric nonnegative zero-diagonal matrix $C_0$, and a scalar $g$. It sets
\[
\alpha=\frac{(1-\epsilon)\operatorname{sigmoid}(g)}
{\lambda_{\max}(C_0)},\qquad
Q=D^{1/2}(I-\alpha C_0)D^{1/2}.
\]
For a symmetric nonnegative matrix, Perron--Frobenius gives spectral radius
$\lambda_{\max}(C_0)$. Hence all eigenvalues of $\alpha C_0$ have magnitude
strictly below one, so $I-\alpha C_0\succ0$. Congruence by $D^{1/2}$ preserves
positive definiteness, and the off-diagonals of $Q$ are nonpositive. Thus $Q$
is an M-matrix for every network weight vector and context. A stricter
diagonally dominant construction is retained as an ablation.

\begin{figure}[h]\centering
\includegraphics[width=\columnwidth]{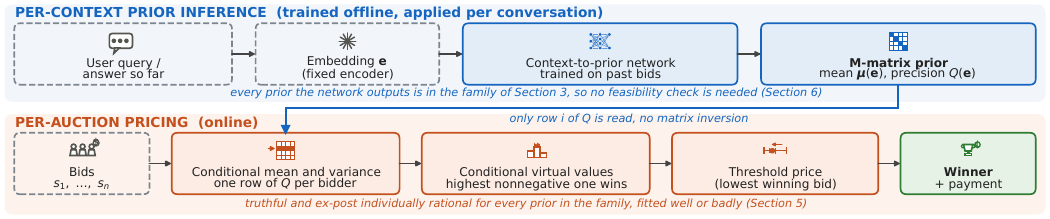}
\caption{Deployment pipeline. \textbf{Offline} (top lane): historical bids and
query embeddings train the context-to-prior map; the spectral parameterization
makes \emph{every} network output an M-matrix, so truthfulness never depends on
fit quality (main text, misspecification-robustness theorem). \textbf{Online}
(bottom lane): one $O(n^2)$ row-read of the precision matrix plus a $60$-step
bisection prices each auction in $\sim$$12\,\mu$s at $n{=}2$ (pricing only;
p50 $0.3$\,ms end-to-end with the network forward pass, Section~\ref{supp:rerun}).}
\end{figure}

\begin{stheorem}[Completeness of the spectral representation]
Let $Q$ be any symmetric positive-definite Z-matrix. Define
\[
D=\operatorname{diag}(Q),\qquad B=I-D^{-1/2}QD^{-1/2}.
\]
Then $D\succ0$, $B$ is symmetric, nonnegative and zero-diagonal,
$\rho(B)<1$, and $Q=D^{1/2}(I-B)D^{1/2}$. Conversely, any positive diagonal
$D$ and symmetric nonnegative zero-diagonal $B$ with $\rho(B)<1$ produce an
SPD Z-matrix through this formula.
\end{stheorem}
\begin{proof}
Positive definiteness gives $Q_{ii}>0$. Directly,
$B_{ii}=0$ and $B_{ij}=-Q_{ij}/\sqrt{Q_{ii}Q_{jj}}\ge0$ for $i\ne j$.
Moreover, $I-B=D^{-1/2}QD^{-1/2}\succ0$, hence every eigenvalue of $B$ is
strictly below one. Since $B$ is nonnegative, Perron--Frobenius identifies its
spectral radius with its largest eigenvalue, so $\rho(B)<1$. The factorization
is the definition rearranged. Conversely, symmetry and $\rho(B)<1$ imply
$I-B\succ0$; congruence by $D^{1/2}$ preserves positive definiteness, while
the signs of the off-diagonals are nonpositive.
\end{proof}

The theorem concerns the abstract factorization. The implementation uses
entrywise softplus weights, diagonal floor $D_{ii}>10^{-3}$, and strength
$(1-\epsilon)\operatorname{sigmoid}(g)$ with $\epsilon=10^{-3}$, so its
realized set is
\[
\begin{aligned}
\mathcal{S}_\epsilon=\{D^{1/2}(I-B)D^{1/2}:\ &D_{ii}>10^{-3},\ B>0\ \text{entrywise off-diagonal},\\
&\rho(B)<1-\epsilon\},
\end{aligned}
\]
a strict subset of the cone. Four precise statements replace any ``entire class'' reading. (i) Every SPD M-matrix with $Q_{ii}>10^{-3}$ for all $i$, normalized coupling
$\rho(B)<1-\epsilon$, and strictly negative off-diagonals is realized exactly. (i$'$) Since $Q_{ii}=D_{ii}$ in this factorization, the diagonal floor excludes matrices with a conditional precision below $10^{-3}$, i.e.\ a conditional scale above $1/\sqrt{10^{-3}}=31.6$ cents, a hard cap on how diffuse a conditional can be rather than a spectral sliver; the fitted seed-$0$ prior touches it ($\min_i\sqrt{Q_{ii}}=0.032$). Truthfulness is unaffected; the coverage claim is ``the cone with $\rho(B)<1-\epsilon$ and $Q_{ii}>10^{-3}$'', not ``the cone up to an $\epsilon$-sliver''.
(ii) M-matrices with zero off-diagonal entries (including diagonal $Q$) are
not realized exactly but lie in the closure of $\mathcal{S}_\epsilon$, since
softplus weights can be made arbitrarily small. (iii) M-matrices with $\rho(B)\in(1-\epsilon,1)$, within $\epsilon$ of singularity, are excluded,
by design: the margin doubles as a conditioning guarantee for the online
inversion-free pricing. Both exclusions, the spectral sliver and the diagonal floor, are
dimension-independent; in particular no dimension-dependent \emph{partial-correlation} budget of the diagonally dominant
$1/(n-1)$ kind applies at any $n$, which is the property the main text uses. (Strict diagonal dominance budgets partial correlations; marginal correlations can still approach $1$, but only at vanishing dominance margin.)

\paragraph{Relaxed-model audit.}
We also fit an SPD precision-Cholesky model and repair it by clipping positive
off-diagonals, symmetrizing, and shifting the diagonal when the minimum
eigenvalue is too small. This is not an exact nearest-cone projection. On the
four-advertiser benchmark the repair is material and reserve-relative revenue
is seed-unstable; these diagnostics are why the paper deploys the structured
model instead.

\paragraph{Four priors on the designed $n=4$ data (main text, Section~7.2).}
Same split ($6{:}1{:}3$ over $780$ queries), five optimization seeds, test-set statistics; $R^2$ is the mean over the four brands of the context-mean prediction on unseen queries, $\bar\rho$ the mean fitted correlation over all six brand pairs, $\rho_{01}$ the fitted correlation of the flagship--premium pair, and ``outside cone'' the share of positive off-diagonal precision entries before any repair (\texttt{fit\_four\_priors\_n4b.py}, \texttt{fit\_four\_priors\allowbreak\_n4b.json}).

\begin{table*}[t]\centering\small
\begin{tabular}{lccccc}
\toprule
Prior & Test NLL & $R^2$ & $\bar\rho$ & $\rho_{01}$ & Outside cone \\
\midrule
A: global closed form (no context) & $18.56$ & $-0.00$ & $0.20$ & $0.67$ & $33\%$ \\
B: unconstrained precision Cholesky & $20.75\pm0.46$ & $-0.03\pm0.02$ & $0.47\pm0.06$ & $0.61\pm0.03$ & $25\pm3\%$ \\
C: diagonally dominant M-matrix & $19.00\pm0.62$ & $+0.60\pm0.06$ & $0.22\pm0.02$ & $0.23\pm0.03$ & $0$ \\
S: spectral M-matrix (deployed) & $18.37\pm0.57$ & $+0.01\pm0.02$ & $0.47\pm0.08$ & $0.58\pm0.05$ & $0$ \\
\bottomrule
\end{tabular}
\end{table*}

The global fit A is included as the no-context reference: its $\rho_{01}$ of $0.67$ is the raw correlation of the two brands' bids, which the context-dependent fits redistribute between the context mean and the residual correlation. C is the only prior that learns an informative context mean, and it does so by giving up correlation; S keeps the correlation and, as Section~7.6 of the main text shows on real queries, that is the channel revenue rides.

\section{Revenue stability under prior misspecification: proof and numerical certification}\label{supp:stab}

\paragraph{Why the conditional coordinates, and a heuristic bridge to KL (moved from the main text).}
The ruler in Proposition~5.8 is the conditional coordinate, not $\|\hat Q-Q\|$: with $\sigma_{i\mid-i}=Q_{ii}^{-1/2}$, revenue is Lipschitz in the precision matrix only locally, with a constant that grows like $Q_{ii}^{-3/2}$ and is useless near the fitted priors (\S\ref{supp:stab-metric}). For Gaussians the conditional-coordinate error is, to leading order, the Fisher metric, which suggests $\E\Delta\lesssim\sqrt{\mathrm{KL}}$, the quantity maximum likelihood controls, and the one for which $\mathrm{MTP}_2$-constrained estimation is unusually well behaved (Lauritzen, Uhler and Zwiernik 2019). We record this as a heuristic bridge, not a theorem: a KL \emph{rate} for the constrained MLE is what would close the chain sample size $\to$ KL error $\to$ revenue loss.

The misspecification-robustness theorem of the main text guarantees that truthfulness survives \emph{any} prior in the M-matrix cone. This section quantifies what misspecification costs in \emph{revenue}, proves the main text's revenue-stability proposition for arbitrary $n$, and certifies each step numerically on the fitted seed-$0$ prior and the held-out bids.

\paragraph{Roadmap.}
Six lemmas, in the order they are used.
Lemma~\ref{lem:mills} is Mills-ratio calculus and fixes the constant $C(\cdot)$ that bounds how fast a virtual value can move when its conditional moments move.
Lemma~\ref{lem:active} shows that being \emph{active} ($\psi_i\ge0$) by itself bounds a bidder's standardized residual from below, so the constant is a property of the prior class rather than a measured quantity, which is what keeps it finite.
Lemma~\ref{lem:slack} shows that moving a bidder's conditional moments from one prior to the other, scale first and mean second, keeps its residual within an explicit slack of the pinned value; this is what lets a bidder that is active under only \emph{one} of the two priors be controlled.
Lemma~\ref{lem:gap} converts a parameter perturbation into a uniform perturbation of the threshold equation's gap function, and is where the clamp $\max\{0,\cdot\}$ earns its keep by excluding inactive bidders.
Lemma~\ref{lem:sandwich} converts that into a bound on the threshold itself, using only the fact that the gap function rises at unit rate; it replaces the implicit function theorem and is what makes the result hold at every $n$.
Lemma~\ref{lem:boundary} handles the profiles where the two priors disagree about who wins, or about whether to sell at all.
Proposition~\ref{prop:supp} assembles them.

\paragraph{Two facts that shape the result.}
First, the natural metric is not the precision matrix: since $\sigma_{i\mid-i}=Q_{ii}^{-1/2}$, revenue is not Lipschitz in $\|\hat Q-Q\|$ near singular priors (\S\ref{supp:stab-metric}). Second, the bound is Lipschitz only up to a boundary term: sale and winner decisions flip, so pointwise payments jump, and a bounded signal density is what turns those jumps into a controlled contribution.

\subsection{Setting and notation}
Bids are truthful for every prior in the cone (main text, misspecification-robustness theorem), so the true law $P^\star$, supported on $[0,\bar s]^n$ with bounded density, is fixed while the prior varies. Fix a profile $\bs$ with provisional winner $w$. For a prior $\theta=(\bmu,Q)$ write
\[
\begin{aligned}
  &\sigma_i=Q_{ii}^{-1/2},\qquad
  \beta_{ij}=-Q_{ij}/Q_{ii}\ \ge 0,\\
  &m_i(t)=\mu_i+\beta_{iw}(t-\mu_w)+\!\!\sum_{j\neq i,w}\!\!\beta_{ij}(s_j-\mu_j),
\end{aligned}
\]
so $m_i(t)$ is bidder $i$'s conditional mean when $w$'s report is set to $t$; $m_w$ does not depend on $t$. Put $z_i(t)=(s_i-m_i(t))/\sigma_i$ and $\psi_i(t)=s_i-\sigma_iR(z_i(t))$, with $\psi_w(t)=t-\sigma_wR\bigl((t-m_w)/\sigma_w\bigr)$. The threshold solves
\[
  D_\theta(t)\;\triangleq\;\psi^\theta_w(t)-\max\Bigl\{0,\ \max_{j\neq w}\psi^\theta_j(t)\Bigr\}\;=\;0,
\]
and $\mathrm{Rev}(\theta)=\E_{P^\star}\bigl[t_0(\bs;\theta)\,\1\{\max_i\psi^\theta_i\ge0\}\bigr]$. Call bidder $j$ \emph{active at $t$} under $\theta$ if $\psi^\theta_j(t)\ge0$; otherwise the clamp removes it. With $\hat\theta$ the second prior,
\[
  \Delta(\bs)=\sum_i\Delta_i(\bs),\qquad
  \Delta_i(\bs)=|\sigma_i-\hat\sigma_i|+\sup_{t\in[0,s_w]}|m_i(t)-\hat m_i(t)|,
\]
the supremum over the bracket $[0,s_w]$ being attained at an endpoint because $m_i$ is affine in $t$; $\|\Delta\|_\infty=\sup_{\bs}\Delta(\bs)<\infty$ by compact support. Finally $\sigma_{\min}=\min_i\min\{\sigma_i,\hat\sigma_i\}$ and $\kappa$ is a bound on the conditional mean-to-scale ratio, $m_i(\bs)/\sigma_i\le\kappa$ at every profile in the support, for both priors, an empirical regularity of the fitted priors ($\kappa=16.93$ on the seed-$0$ prior and held-out bids, \S\ref{sub:cert}), not an assumption of the main text. Since $m_w$ is free of $t$ and every $m_j(t)$ is nondecreasing in $t$ ($\beta_{jw}\ge0$), the bound at the realized profile covers every $t\in[0,s_w]$, which is all the proof uses.

\begin{sremark}[Why the supremum over $t$]
Evaluating the conditional means at the realized profile only is not enough: the competitor's conditional mean enters the threshold equation at $t_0$, not at $s_w$, and when the two priors disagree on a regression weight $\beta$ the discrepancy grows like $|\beta-\hat\beta|\,|s_w-t_0|$, which the realized-profile version does not cover. Taking the supremum over the bracket $[0,s_w]$ closes this and costs nothing computationally.
\end{sremark}

\subsection{Six lemmas}

\begin{slemma}[Mills-ratio calculus and the constant]\label{lem:mills}
$R'(z)=zR(z)-1<0$ and $zR(z)<1$ for all $z$, hence
\[
  \frac{\partial\psi_i}{\partial s_i}=2-z_iR(z_i)>1,\qquad
  \frac{\partial\psi_i}{\partial m_i}=R'(z_i),\qquad
  \frac{\partial\psi_i}{\partial\sigma_i}=-R(z_i)+z_iR'(z_i).
\]
Write $g(z)=R(z)-zR'(z)=R(z)(1-z^2)+z$. Then
\begin{enumerate}
\item[(a)] $z\mapsto|R'(z)|=1-zR(z)$ is strictly decreasing, so $\sup_{z\ge z_0}|R'|=1-z_0R(z_0)$;
\item[(b)] $g'(z)=z\bigl[z-R(z)(1+z^2)\bigr]$ has the sign of $-z$, so $g$ is unimodal with maximum $g(0)=\sqrt{\pi/2}$, and $g(z)\to0^+$ as $z\to+\infty$, $g(z)\to-\infty$ as $z\to-\infty$;
\item[(c)] consequently, with
  \[
    C(z_0)\triangleq\max\bigl\{1-z_0R(z_0),\ |g(z_0)|,\ \sqrt{\pi/2}\bigr\},
  \]
  one has $\sup_{z\ge z_0}\max\{|R'(z)|,|g(z)|\}\le C(z_0)$, with equality whenever $z_0\le0$; $C$ is nonincreasing and $C(z_0)\to\infty$ as $z_0\to-\infty$.
\end{enumerate}
\end{slemma}

\begin{proof}
Differentiating $(1-\Phi)/\phi$ with $\phi'=-z\phi$ gives $R'=(-\phi^2+z\phi(1-\Phi))/\phi^2=zR-1$; $zR<1$ is the classical Mills bound. The three partials follow by the chain rule with $\partial z_i/\partial m_i=-1/\sigma_i$ and $\partial z_i/\partial\sigma_i=-z_i/\sigma_i$.
(a) $\frac{d}{dz}(1-zR)=z-R(1+z^2)<0$ by the Gordon bound $R(z)>z/(1+z^2)$, trivial for $z\le0$.
(b) $g'=R'(1-z^2)-2zR+1=z[z-R(1+z^2)]$, whose bracket is negative by the same bound; $g(0)=R(0)=\sqrt{\pi/2}$. The tails use $R(z)=z^{-1}-z^{-3}+O(z^{-5})$ at $+\infty$, giving $g(z)=2/z+O(z^{-3})$, and $R(z)\sim\sqrt{2\pi}e^{z^2/2}$ at $-\infty$ against $(1-z^2)\to-\infty$.
(c) On $[z_0,\infty)$ the supremum of $g$ is $\sqrt{\pi/2}$ when $z_0\le0$ and $g(z_0)<\sqrt{\pi/2}$ otherwise; its infimum is $\min\{g(z_0),0\}$. Monotonicity of $C$ follows from (a), (b).
\end{proof}

\begin{sremark}[All three terms of $C$ are needed]
Two of the three terms are easy to drop. The absolute value on $g$ is needed because $g<0$ for $z_0<-0.840$ and $|g|$ overtakes $|R'|$ once $z_0<-1.586$. The interior term $\sqrt{\pi/2}$ is needed whenever $z_0\le0$, since $g$ peaks at $z=0$ rather than at the left endpoint. At the residual floor measured at the threshold in our data ($-1.297$) the first term binds and the correction is invisible; at the derived floor $z_0=h^{-1}(\kappa)=-1.908$ (\S\ref{sub:cert}) the second term binds, and the corrected constant $41.6$ replaces the $29.7$ an uncorrected formula would print.
\end{sremark}

\begin{slemma}[Activity pins the residual]\label{lem:active}
Let $h(z)=R(z)-z$; then $h'(z)=zR(z)-2<0$, so $h$ is a strictly decreasing bijection of $\R$. If bidder $i$ is active at $t$ under $\theta$, i.e.\ $\psi^\theta_i(t)\ge0$, then
\[
  h\bigl(z^\theta_i(t)\bigr)\le\frac{m^\theta_i(t)}{\sigma_i},
  \quad\text{i.e.}\quad
  z^\theta_i(t)\ \ge\ h^{-1}\!\Bigl(\frac{m^\theta_i(t)}{\sigma_i}\Bigr)\ \ge\ h^{-1}(\kappa).
\]
In particular $z_0\triangleq\min\{h^{-1}(\kappa),0\}$ lower-bounds the residual of every bidder at every point where it is active, under either prior; this covers the winner at any $t\ge t_0$ (where $\psi_w(t)\ge\psi_w(t_0)=\max\{0,\cdot\}\ge0$) and every active competitor.
\end{slemma}

\begin{proof}
Substituting $s_i=m_i+\sigma_iz_i$ into $\psi_i=s_i-\sigma_iR(z_i)\ge0$ gives $m_i/\sigma_i\ge R(z_i)-z_i=h(z_i)$; invert the decreasing $h$ and use $m_i/\sigma_i\le\kappa$. Monotonicity of $h$ is the Mills bound again: $h'=R'-1=zR-2<0$; $h(z)\to\pm\infty$ as $z\to\mp\infty$.
\end{proof}

This replaces a measured quantity by a derived one: $z_0$ follows from the single scalar $\kappa$. The bound is the companion, from above, of the main text's positive-mean-dominance regularity, which constrains the same ratio from below. It also yields a consistency check reported in \S\ref{sub:cert}. What Lemma~\ref{lem:active} does \emph{not} give is a bound on the residual of a bidder under a prior where it is \emph{inactive}; the next lemma supplies it.

\begin{slemma}[Residual slack along the two-step path]\label{lem:slack}
Fix a bidder $i$, a point $t$, and suppose $z^\theta_i(t)\ge z_0$ with $z_0\le0$ (for instance because $i$ is active at $t$ under $\theta$). Move the parameters from $(m_i,\sigma_i)$ to $(\hat m_i,\hat\sigma_i)$ in two monotone steps, \emph{scale first, then mean}: $(m_i,\sigma_i)\to(m_i,\hat\sigma_i)\to(\hat m_i,\hat\sigma_i)$. Every intermediate residual satisfies
\[
\begin{aligned}
  z\ &\ge\ z_0-\frac{|z_0|\,|\sigma_i-\hat\sigma_i|+|m_i(t)-\hat m_i(t)|}{\sigma_{\min}}\\
   &\ge\ z_1\triangleq z_0-\frac{\max\{1,|z_0|\}\,\|\Delta\|_\infty}{\sigma_{\min}},
\end{aligned}
\]
and consequently, by Lemma~\ref{lem:mills}(c),
\[
  \bigl|\psi^\theta_i(t)-\psi^{\hat\theta}_i(t)\bigr|\le C(z_1)\,\Delta_i(\bs).
\]
The same holds with the roles of $\theta$ and $\hat\theta$ exchanged.
\end{slemma}

\begin{proof}
Write $z=z^\theta_i(t)$, $\delta=s_i-m_i(t)=z\sigma_i$. \emph{Step 1} ($\sigma_\lambda$ from $\sigma_i$ to $\hat\sigma_i$, mean fixed): the residual is $\delta/\sigma_\lambda=z\,\sigma_i/\sigma_\lambda$. If $z\ge0$ it stays nonnegative, hence $\ge z_0$. If $z<0$ it is smallest when $\sigma_\lambda$ is smallest, and $z\sigma_i/\sigma_{\min}=z+z(\sigma_i-\sigma_{\min})/\sigma_{\min}\ge z-|z|\,|\sigma_i-\hat\sigma_i|/\sigma_{\min}\ge z_0-|z_0|\,|\sigma_i-\hat\sigma_i|/\sigma_{\min}$, using $z_0\le z<0$. \emph{Step 2} (mean from $m_i(t)$ to $\hat m_i(t)$ at scale $\hat\sigma_i$): the residual changes monotonically by at most $|m_i(t)-\hat m_i(t)|/\hat\sigma_i\le|m_i(t)-\hat m_i(t)|/\sigma_{\min}$ from its value at the end of step~1, which is $\ge z_0-|z_0||\sigma_i-\hat\sigma_i|/\sigma_{\min}$ by the same argument (that value is either $\ge0$ or the step-1 minimum). Combining gives the first display; bounding $|m_i(t)-\hat m_i(t)|+|\sigma_i-\hat\sigma_i|\le\Delta_i(\bs)\le\|\Delta\|_\infty$ gives $z_1$. Along the path both parameter sensitivities of $\psi_i$ are bounded by $C(z_1)$ (Lemma~\ref{lem:mills}(c) with $z_1\le z_0\le0$), so integrating $d\psi_i=(\partial\psi_i/\partial\sigma)\,d\sigma+(\partial\psi_i/\partial m)\,dm$ along the two steps gives $|\psi^\theta_i(t)-\psi^{\hat\theta}_i(t)|\le C(z_1)\bigl(|\sigma_i-\hat\sigma_i|+|m_i(t)-\hat m_i(t)|\bigr)\le C(z_1)\Delta_i(\bs)$.
\end{proof}

\begin{sremark}[Order of the moves matters]
Moving the mean first and the scale second would pass through the point $(\hat m_i,\sigma_i)$, whose residual $(s_i-\hat m_i)/\sigma_i=z^{\hat\theta}_i\,\hat\sigma_i/\sigma_i$ is a residual under neither prior and, when $\hat\sigma_i>\sigma_i$ and $z^{\hat\theta}_i<0$, lies below both endpoint residuals by a factor rather than an additive slack. Scale first keeps the excursion additive and first-order in $\Delta$. The statement ``every intermediate residual lies between the endpoint values'' is false for the mean-first order.
\end{sremark}

\begin{slemma}[Uniform perturbation of the gap function]\label{lem:gap}
At every $t\in[0,s_w]$ at which the winner $w$ is active under at least one of the two priors,
\[
  \bigl|D_\theta(t)-D_{\hat\theta}(t)\bigr|\;\le\;C(z_1)\,\Delta(\bs)\;\triangleq\;\eta(\bs).
\]
\end{slemma}

\begin{proof}
Write $a_j=\psi^\theta_j(t)$, $\hat a_j=\psi^{\hat\theta}_j(t)$ for $j\ne w$, $A=\max_ja_j$, $\hat A=\max_j\hat a_j$, so $D_\theta-D_{\hat\theta}=(a_w-\hat a_w)-(\max\{0,A\}-\max\{0,\hat A\})$.

\emph{Winner.} $w$ is active at $t$ under one prior, so Lemma~\ref{lem:active} pins its residual there and Lemma~\ref{lem:slack} (started from that prior) gives $|a_w-\hat a_w|\le C(z_1)\Delta_w(\bs)$.

\emph{Only active competitors enter.} If $A\le0$ then $\max\{0,A\}-\max\{0,\hat A\}\le0$. If $A>0$ with argmax $j^\star$, then $\max\{0,A\}=a_{j^\star}$ and $\max\{0,\hat A\}\ge\hat a_{j^\star}$, so $\max\{0,A\}-\max\{0,\hat A\}\le a_{j^\star}-\hat a_{j^\star}$, where $j^\star$ is active at $t$ under $\theta$; Lemmas~\ref{lem:active} and~\ref{lem:slack} give $a_{j^\star}-\hat a_{j^\star}\le C(z_1)\Delta_{j^\star}(\bs)$. Symmetrically, $\max\{0,\hat A\}-\max\{0,A\}\le\hat a_{\hat j}-a_{\hat j}\le C(z_1)\Delta_{\hat j}(\bs)$ with $\hat j$ the $\hat\theta$-argmax, active under $\hat\theta$, or $\le0$ if $\hat A\le0$. Bidders inactive under both priors never appear.

Adding the winner's term to whichever competitor term applies and using $\Delta_w+\Delta_j\le\Delta$ gives the claim.
\end{proof}

\begin{sremark}[No combinatorial factor]
A factor $2$ ``for winner and active loser'', or $1+\#\{\text{active}\}$ at general $n$, would double count: $\Delta(\bs)$ already sums over bidders, and at any $t$ the gap function involves the winner plus whichever single competitor attains the inner maximum. The constant in \S\ref{sub:cert} is correspondingly $13.6$, not $21.3$.
\end{sremark}

\begin{slemma}[Threshold stability, sandwich form]\label{lem:sandwich}
$D_\theta$ is continuous and strictly increasing, and $D_\theta(t)-D_\theta(t_0)\ge t-t_0$ for $t\ge t_0$. Consequently, if $|D_\theta(t)-D_{\hat\theta}(t)|\le\eta$ at $t=t_0(\theta)+\eta$ and at $t=t_0(\hat\theta)+\eta$, then $|t_0(\theta)-t_0(\hat\theta)|\le\eta$.
\end{slemma}

\begin{proof}
$\partial\psi_w/\partial t=2-z_wR(z_w)>1$ by Lemma~\ref{lem:mills}, while each $t\mapsto\psi_j(t)$ is nonincreasing because $\partial\psi_j/\partial t=\beta_{jw}R'(z_j)\le0$ with $\beta_{jw}\ge0$ and $R'<0$; the clamp and the maximum preserve this. So $D_\theta$ rises at rate at least one, which is the displayed inequality. Then $D_{\hat\theta}(t_0+\eta)\ge D_\theta(t_0+\eta)-\eta\ge\eta-\eta=0$, and since $D_{\hat\theta}$ is strictly increasing with unique zero $\hat t_0$, $\hat t_0\le t_0+\eta$. Exchanging the roles of the priors at $t=\hat t_0+\eta$ gives $t_0\le\hat t_0+\eta$.
\end{proof}

\begin{sremark}[What this buys]
This is the step that replaces the implicit function theorem. It requires no differentiability of $\max_j$ (which fails wherever the inner argmax switches) and no interpolation path inside the cone. The latter matters: linear interpolation of $(\beta_i,\sigma_i)$ violates the symmetry $Q_{ij}=Q_{ji}$, so the interpolated ``prior'' need not be admissible. The argument uses only $\sigma>0$ and $\beta\ge0$, which the coordinatewise moves of Lemma~\ref{lem:slack} preserve, and it is indifferent to $n$. Note also that Lemma~\ref{lem:gap} is invoked only at $t_0+\eta$ and $\hat t_0+\eta$, points at which the winner is active under the corresponding prior, so the winner's residual is never needed off its own active region.
\end{sremark}

\begin{slemma}[Boundary term]\label{lem:boundary}
Let $\eta_{\max}=C(z_1)\|\Delta\|_\infty$. The sale sets and winner sets under $\theta$ and $\hat\theta$ differ only on
\[
  \bigl\{|\max_i\psi^\theta_i|\le\eta_{\max}\bigr\}\cup
  \bigl\{|\psi^\theta_{(1)}-\psi^\theta_{(2)}|\le\eta_{\max}\bigr\}
\]
(virtual values at the realized profile). If $f_s,f_w$ bound the densities of $\max_i\psi^\theta_i$ and of $\psi^\theta_{(1)}-\psi^\theta_{(2)}$ uniformly on $[-\eta_{\max},\eta_{\max}]$ (a bound at $0$ alone does not control a strip of finite width), and $\bar t_s,\bar t_w$ bound the conditional mean of the larger of the two payments on the respective strip, the disagreement set contributes at most $2(\bar t_sf_s+\bar t_wf_w)\,C(z_1)\|\Delta\|_\infty$ to $|\mathrm{Rev}(\theta)-\mathrm{Rev}(\hat\theta)|$.
\end{slemma}

\begin{proof}
If $\max_i\psi^\theta_i>\eta_{\max}$, its argmax $i$ is active under $\theta$, so Lemmas~\ref{lem:active}--\ref{lem:slack} give $\psi^{\hat\theta}_i\ge\psi^\theta_i-C(z_1)\Delta_i>0$: both priors sell. If $\max_i\psi^\theta_i<-\eta_{\max}$ and some $j$ were active under $\hat\theta$, the same lemmas started from $\hat\theta$ would give $\psi^\theta_j\ge\psi^{\hat\theta}_j-C(z_1)\Delta_j\ge-\eta_{\max}$, a contradiction; so neither sells. For the winner, suppose both sell and $\psi^\theta_{(1)}-\psi^\theta_{(2)}>\eta_{\max}$ with $\theta$-winner $i$: then $\psi^{\hat\theta}_i\ge\psi^\theta_i-C(z_1)\Delta_i$, and any $j\ne i$ active under $\hat\theta$ has $\psi^{\hat\theta}_j\le\psi^\theta_j+C(z_1)\Delta_j\le\psi^\theta_{(2)}+C(z_1)\Delta_j$, so $\psi^{\hat\theta}_i-\psi^{\hat\theta}_j\ge(\psi^\theta_{(1)}-\psi^\theta_{(2)})-C(z_1)(\Delta_i+\Delta_j)>0$; bidders inactive under $\hat\theta$ cannot win. The probability of each strip is at most $2f\eta_{\max}$, and on it the revenue difference is at most the larger payment, whose conditional mean is bounded by $\bar t$; payments never exceed the winner's own report (Theorem~\ref{supp:thm:threshold}), which is bounded on the support.
\end{proof}

\begin{sremark}[Why this term is $L^\infty$, not $L^1$]
Applying the strip bound with the \emph{random} radius $\eta(\bs)$ and then taking expectations would be inadmissible: the density bound $\Pr(\text{flip})\le2f\delta$ requires a deterministic $\delta$, and $\|\Delta\|_\infty=\infty$ under unbounded support whenever $\beta_i\neq\hat\beta_i$, because $m_i(\cdot)$ is affine in the opponent's report. Compact support, which is the deployed regime since bids are bounded in cents, repairs the step. Under unbounded support, truncating at radius $M$ and optimizing $M\asymp\sqrt{2\log(1/\varepsilon)}$ with $\varepsilon=\E[\Delta]$ gives $O\bigl(\varepsilon\sqrt{\log(1/\varepsilon)}\bigr)$: still near-linear, at the cost of one logarithmic factor.
\end{sremark}

\subsection{The proposition}

\begin{sproposition}[Revenue stability; arbitrary $n$]\label{prop:supp}
Let $\theta,\hat\theta$ lie in the cone with conditional mean-to-scale ratios bounded by $\kappa$ on the support, let $P^\star$ be supported on $[0,\bar s]^n$ with bounded density, and let $z_0=\min\{h^{-1}(\kappa),0\}$, $z_1=z_0-\max\{1,|z_0|\}\|\Delta\|_\infty/\sigma_{\min}$. Then
\[
  \bigl|\mathrm{Rev}(\theta)-\mathrm{Rev}(\hat\theta)\bigr|
  \le C(z_1)\,\E_{P^\star}\bigl[\Delta(\bs)\bigr]
   +2\bigl(\bar t_sf_s+\bar t_wf_w\bigr)C(z_1)\,\|\Delta\|_\infty .
\]
As $\|\Delta\|_\infty\to0$, $z_1\to z_0$ and $C(z_1)\to C(z_0)$: revenue is locally Lipschitz in the conditional coordinates with sharp constant $C(h^{-1}(\kappa))$.
\end{sproposition}

\begin{proof}
Split on the event $A$ that the sale decision and the winner identity agree under the two priors. On $A$, the winner $w$ is common, it is active under $\theta$ at $t_0+\eta$ and under $\hat\theta$ at $\hat t_0+\eta$, so Lemma~\ref{lem:gap} applies at both points (when $t_0+\eta\ge s_w$ the inequality $\hat t_0\le t_0+\eta$ is immediate from $\hat t_0\le s_w$, Theorem~\ref{supp:thm:threshold}, and likewise with the priors exchanged) and Lemma~\ref{lem:sandwich} gives $|t_0(\theta)-t_0(\hat\theta)|\le C(z_1)\Delta(\bs)$ pointwise; taking expectations yields the first term. On $A^c$, Lemma~\ref{lem:boundary} yields the second. Continuity of $C$ gives the limit statement.
\end{proof}

\begin{sremark}[What is derived and what is measured]
Every ingredient of the constant is now read off the prior class and the perturbation size: $\kappa$ and $\sigma_{\min}$ from the priors, $\|\Delta\|_\infty$ from the pair. Nothing about the realized residuals is assumed. The price of this is the slack $z_0\to z_1$, which is first order in $\|\Delta\|_\infty/\sigma_{\min}$ and vanishes in the Lipschitz limit; quoting $C(z_0)$ for a finite perturbation would be correct only for bidders active under \emph{both} priors and for the winner exactly at its threshold. The certification below reports both $C(z_0)$, the limit constant, and the realized residual floor.
\end{sremark}

\subsection{Why not $\|\hat Q-Q\|$}\label{supp:stab-metric}
Locally around any fixed SPD prior, revenue \emph{is} Lipschitz in the precision matrix: $Q\mapsto Q_{ii}^{-1/2}$ and $Q\mapsto-Q_{ij}/Q_{ii}$ are smooth on the SPD interior, and composing them with Proposition~\ref{prop:supp} gives a constant of order $Q_{ii}^{-3/2}$. What fails is a \emph{useful} statement $\mathrm{Rev}(\hat Q)\ge\mathrm{Rev}(Q)-L\|\hat Q-Q\|$ with one $L$ for the class: the constant explodes as $Q_{ii}\to0$, and at the fitted priors it is already hopeless. The seed-$0$ spectral prior has $\min_i\sqrt{Q_{ii}}=0.032$ (median $0.064$): conditional standard deviations reach $\approx31$ cents, and $\sigma=Q_{ii}^{-1/2}$ is violently nonlinear there. Perturbing $Q$ additively in operator norm and replaying on the held-out bids (script \texttt{regret\_probe2.py}), the smallest perturbation already moves revenue by $9$ cents out of $37.9$; the empirical ratio $|\Delta\mathrm{Rev}|/\|\Delta Q\|$ is $220$--$650$ (far below the local constant, as it must be, and far above anything a platform could use), and neither a linear nor a quadratic fit has positive $R^2$ at the perturbation sizes tried (a large finite ratio is evidence of bad conditioning, not of non-Lipschitzness). Restricting to a sub-cone $Q_{ii}\ge c$ does not repair it ($c=0.1$: ratio $223$). The conditional coordinates are the right ruler, since they are exactly what the precision-conditioning lemma reads off row $i$ of $Q$, and for Gaussians their squared error is, to leading order, the Fisher metric, which suggests $\E\Delta\lesssim\sqrt{\mathrm{KL}}$ and so stability in $\sqrt{\mathrm{KL}}$: the quantity maximum likelihood controls, and the one under which $\mathrm{MTP}_2$-constrained estimation is unusually well behaved (Lauritzen, Uhler and Zwiernik, 2019). We record this as a heuristic bridge rather than a theorem. What is missing is a KL \emph{rate} for the constrained MLE, of which the cited work gives existence under very few samples; supplying it would close the chain sample size $\to$ KL error $\to$ revenue loss and make the class learnable in the revenue sense, not only in the sense that its MLE is convex.

\subsection{Numerical certification}\label{sub:cert}
{\sloppy Fitted seed-$0$ spectral prior, $n{=}2$ held-out bids ($180$ queries, $3{,}600$ auctions, $3{,}241$ sold); scripts \texttt{prop\_verify.py}, \texttt{prop\_verify2.py}, \texttt{regret\_probe3.py}, and the lemma-level diagnostics \texttt{mills\_checks.py} (Lemma~\ref{lem:mills}), \texttt{check\_z0.py} (Lemma~\ref{lem:active}, elementwise) and \texttt{check\_stability.py} (Lemmas~\ref{lem:active} and \ref{lem:gap}--\ref{lem:boundary} on a pair of priors), driven by \texttt{run\_certification.py}; all in \texttt{code/n4}, numbers read from their JSON outputs and from \texttt{prop\_verify2.json} and \texttt{regret\_probe3\_results.json}. The perturbed prior follows the relative scheme of \texttt{regret\_probe3.py} (diagonal $\times e^{\varepsilon u}$, off-diagonal $\times(1+\varepsilon u)$ clipped to the cone, means shifted by $\varepsilon u\sigma$); the row quoted uses $\varepsilon=5\times10^{-4}$, the smallest scale, as the Lipschitz-limit regime.\par}

\begin{table*}[t]\centering\footnotesize\setlength{\tabcolsep}{4pt}
\resizebox{\textwidth}{!}{
\begin{tabular}{@{}lll@{}}
\toprule
Lemma / quantity & Check & Result \\
\midrule
\ref{lem:mills}: $R'=zR-1$ & finite differences, $z\in[-3,6]$ & rel.\ error $3\times10^{-8}$ \\
\ref{lem:mills}: $2-zR>1$ & grid $z\in[-6,6]$ & min $1.026$; $\to1$ at $+\infty$ \\
\ref{lem:mills}(a): $|R'|$ decreasing & grid & confirmed \\
\ref{lem:mills}(b): $\max g=\sqrt{\pi/2}$ & grid & $1.2533$ at $z=0$ \\
\ref{lem:active}: $h(z_i)\le m_i/\sigma_i$ & all active bidders, sold auctions & $0/4{,}339$ violations \\
$\kappa=\max_i m_i/\sigma_i$ & test set, fitted prior (median $3.25$) & $16.93$ ($\ge6.54$ \checkmark) \\
$z_0=h^{-1}(\kappa)$, derived & closed form & $-1.908$ \\
$C(z_0)$, derived limit constant & closed form ($|g|$ binds) & $\mathbf{41.6}$ \\
$\sigma_{\min}$; $\|\Delta\|_\infty$ & cone pair, $\varepsilon=5\times10^{-4}$ & $2.89$ cents; $0.40$ cents \\
$z_1$; $C(z_1)$ & same pair & $-2.172$; $99.3$ \\
\ref{lem:gap}: $|D_\theta-D_{\hat\theta}|\le C(z_1)\Delta$ at $t_0{+}\eta,\hat t_0{+}\eta$ & same pair & $0/6{,}476$ violations \\
\ref{lem:sandwich}: $|t_0-\hat t_0|\le\eta$ & same pair, agreement event & $0/3{,}238$ violations \\
\ref{lem:boundary}: flips inside the strips & same pair & $0/3{,}241$ violations \\
$|t_0-\hat t_0|\le C(z_0)\Delta$, \emph{limit} constant & same pair (diagnostic) & $0/3{,}238$ violations \\
measured $|\partial t_0/\partial m|$ & median / max & $0.70$ / $0.90$ \\
measured $|\partial t_0/\partial\sigma|$ & median / max & $0.28$ / $0.75$ \\
active competitors & share of sold auctions & $65.6\%$ (inactive $34.4\%$) \\
residual floor at $t_0$ (winner; active competitor) & min, sold auctions & $-1.291$; $-1.297$ \\
residual floor at the realized profile (active competitor) & min, sold auctions & $-1.834$ \\
$C$ at the threshold floor & closed form at $-1.297$ & $7.80$ \\
\ref{lem:boundary}: $f_s,\bar t_s$ & $|\max\psi|<1$ cent & $0.0081$/cent, $14.8$ cents \\
\ref{lem:boundary}: $f_w,\bar t_w$ & $|\psi_1-\psi_2|<1$ cent & $0.0057$/cent, $43.6$ cents \\
\midrule
aggregate $C\bigl(1+2(\bar t_sf_s+\bar t_wf_w)\bigr)$, derived ($C=41.6$) & uniform bound & $\mathbf{72.2}$ \\
aggregate at the $t_0$ floor ($C=7.80$) & lower bound on the derived one & $13.6$ \\
per-auction $|\partial D/\partial m|$ & mean / median / max & $2.38$ / $2.35$ / $8.83$ \\
empirical $|\Delta\mathrm{Rev}|/$mean coordinate error & cone perturbations, $8$ scales (\texttt{regret\_probe3.py}) & $0.24$ (linear fit $R^2$ $0.93$) \\
empirical vs.\ $\mathrm{KL}$ & same perturbations & $1.65\,\mathrm{KL}$ ($R^2$ $0.94$); $2.36\sqrt{\mathrm{KL}}$ ($0.91$) \\
\bottomrule
\end{tabular}}
\end{table*}

\paragraph{Reading.}
The derived uniform constant ($72.2$) exceeds the measured average slope ($0.24$) by a factor of about $300$, an order-of-magnitude comparison only, since the constant multiplies $\E\Delta$ plus an $\|\Delta\|_\infty$ boundary term while the slope is per unit of mean coordinate error, as a worst-case-over-auctions, worst-direction bound must; the constant one would quote from the residual floor measured at the threshold ($13.6$) is a factor $5$ smaller, and Lemma~\ref{lem:active} explains why it is only a lower bound, since the floor is attained by a bidder at $\psi\approx0$ with $m/\sigma=6.54$, whereas the class contains conditional means up to $16.9$ scales. The analytic sensitivities average $2.4$ against a measured $\partial D/\partial t\approx3.4$, so the typical threshold moves by $\approx0.7$ cent per cent of conditional-mean error: sub-unit, the visible work of $\partial\psi/\partial s>1$. Dropping the active-competitor refinement of Lemma~\ref{lem:gap} admits residuals down to $z=-3.46$, where $C=1.09\times10^4$ and the aggregate constant is $1.9\times10^4$, so the clamp in the threshold equation buys more than two orders of magnitude and is what makes the bound informative at all. At finite perturbation size the slack $z_0\to z_1$ is real: with $\sigma_{\min}=2.9$ cents, a $0.4$-cent perturbation already costs a factor $2.4$ in the constant ($99.3$ vs.\ $41.6$), and by $\|\Delta\|_\infty\approx1.5$ cents the guaranteed constant exceeds $10^3$, while the limit constant $C(z_0)$ still bounds every realized threshold move ($0/3{,}238$ violations at every scale tested). The gap between the two is the price of a proof that assumes nothing about where the perturbed residuals land.

\paragraph{A support bound for $\kappa$.}
The tabulated $\kappa=16.93$ is the maximum of $m_i/\sigma_i$ over realized test profiles. Because $m_i(\bs)=\mu_i+\sum_j\beta_{ij}(s_j-\mu_j)$ has $\beta_{ij}=-Q_{ij}/Q_{ii}\ge0$, its supremum over the box $[0,\bar s]^n$ is attained at the vertex $s_j=\bar s$ for all $j$, so a bound valid on the whole support, not a finite-sample maximum, is computable in closed form from the fitted prior: $\kappa_{\mathrm{box}}=40.0$ on the seed-$0$ test contexts ($\bar s=200$), giving $z_0=h^{-1}(\kappa_{\mathrm{box}})=-2.33$ and $C(z_0)=169$ (\texttt{kappa\_box.py}, \texttt{kappa\_box\_results.json}). The proposition may be read with either constant: the realized one describes the data, the vertex one is a certificate.

\paragraph{Consistency check on the residual floor.}
Lemma~\ref{lem:active} turns the measured residual floor into a test of the fitted priors: a residual of $-1.297$ among active bidders is possible only if $\max_im_i/\sigma_i\ge h(-1.297)=6.54$ somewhere in the test set, so $\kappa\ge6.54$ and the derived constant $C(h^{-1}(\kappa))$ is at least the $7.80$ quoted above, with equality iff the floor is attained by a bidder exactly at $\psi=0$. The realized $\kappa=16.93$ passes this check with room to spare, and the elementwise version of Lemma~\ref{lem:active} holds at every one of the $4{,}339$ active-bidder records (minimum slack $-3.6\times10^{-15}$, i.e.\ equality at the binding records up to rounding). The two easy ways to fail the check (evaluating the winner's residual at $s_w$ instead of $t_0$, or reading $\sigma$ from $\Sigma_{ii}$ rather than $1/Q_{ii}$) are what \texttt{check\_z0.py} guards against. Note also that the competitor's floor depends on where it is evaluated: $-1.297$ at the threshold $t_0$ (where Lemma~\ref{lem:gap} uses it) but $-1.834$ at the realized profile (where Lemma~\ref{lem:boundary} uses it), because $m_j(t)$ is nondecreasing in $t$ and $t_0\le s_w$; the derived $z_0=-1.908$ lies below both, as it must.

\paragraph{Empirical trace.}
The main text's prior-perturbation study is the bound's empirical trace: mean noise $\varepsilon$ perturbs $m_i$ alone, and the rent decays monotonically ($+6.4\%\to+2.5\%\to-4.3\%$ at $\varepsilon=0,10,20$). The successive drops ($-3.9$, $-6.8$) are mildly convex in $\varepsilon$, so the study is \emph{compatible with} the bound rather than a confirmation of a linear law: an upper bound constrains the decay from above only. In the same twelve-regime study, mis-scaling the conditional variance by $\times0.5$ and $\times2$ leaves the rent at $+5.9\%$ and $+2.4\%$ (\texttt{perturb\_out.json}).

\section{Experimental protocol and reproducibility}

\paragraph{Within-query correlations on the four-brand data (main text, Section~7.1).} Bids are demeaned within each query of the four-brand file (\texttt{simulated\_bids\_n4b.pt}; $780$ queries, $15{,}600$ rows per brand) and correlated across all rows. All six pairs are positive; the main-line pair (premium, value) is $0.58$ on this data set and $0.66$ on the $n{=}2$ main-line data (\texttt{ST\_600.pt}, $600$ queries), which differ in query set and scenario count. Script and results: \texttt{pairwise\_within\_query\_corr} (\texttt{.py} and \texttt{\_n4b.json}).

\begin{center}\resizebox{\columnwidth}{!}{\small
\begin{tabular}{lcccc}
\toprule
 & Flagship & Premium & Value & Clearance \\
\midrule
Flagship & $1$ & $0.63$ & $0.28$ & $0.20$ \\
Premium & $0.63$ & $1$ & $0.58$ & $0.43$ \\
Value & $0.28$ & $0.58$ & $1$ & $0.57$ \\
Clearance & $0.20$ & $0.43$ & $0.57$ & $1$ \\
\bottomrule
\end{tabular}
}\end{center}

\paragraph{Prompt-ablation table (main text, Section~7.7).}
Seven gpt-4o-mini prompt schemes and the gpt-4o replication. Gradient $=$ ratio of the premium brand's mean bid on transactional to informational queries (query means first); Corr.\ $=$ within-query correlation of the two brands' bids after removing query means. The last four columns describe the shape of the bid distribution (cents on a $0$--$200$ scale): median, share of bids at or below $10$ cents, 95th percentile, and the share of all bids carried by the three most frequent values. Explicit calculation piles bids at the floor; the judgment-based schemes, facts-only and behavioral alike, spread bids across the scale, and differ in whether the spread follows the query's intent (gradient) or only the brand. Source: \texttt{pilot*.json} (the v3 row uses the full v3 run \texttt{simulated\_bids\_n4\_full.json}, the data its gradient and correlation were computed on), statistics by \texttt{bid\_shape\_by\_scheme.py} (\texttt{code/n4/bid\_shape\_by\_scheme.json}). The v4.1 pilot ($48$ calls) is omitted: its expected values exceeded the $200$-cent cap and the run was stopped.

\begin{table*}[t]\centering\small
\begin{tabular}{lcc|cccc}
\toprule
Prompt scheme & Gradient & Corr. & Median & $\le 10$ & P95 & Top-3 \\
\midrule
Facts only, qualitative (v3) & $0.98$ & $+0.78$ & $30$ & $13\%$ & $94$ & $27\%$ \\
$+$qualitative conversion funnel (v4) & $1.74$ & $+0.74$ & $30$ & $6\%$ & $85$ & $32\%$ \\
$+$quantified conversion facts (v4.2) & $8.56$ & $+0.41$ & $12$ & $46\%$ & $80$ & $25\%$ \\
$+$chain-of-thought (v4.2$+$CoT) & $6.71$ & $+0.06$ & $5$ & $80\%$ & $80$ & $54\%$ \\
$+$quantified demand index$+$decimals (v4.3) & $7.40$ & $+0.09$ & $4$ & $84\%$ & $100$ & $27\%$ \\
$+$low temperature ($T{=}0.2$) & $8.94$ & $-0.03$ & $3$ & $88\%$ & $120$ & $29\%$ \\
Behavioral traits (main-line design) & $5.76$ & $+0.66$ & $25$ & $24\%$ & $157$ & $27\%$ \\
gpt-4o, v4.3 & $14.52$ & $+0.36$ & $1.2$ & $88\%$ & $100$ & $20\%$ \\
\bottomrule
\end{tabular}
\end{table*}

\paragraph{Is the demand signal the only common input? (main text, Sections~7.1 and~7.7).}
Given a query and a shock realization, each brand's bid is an independent LLM call, so any co-movement beyond what the shared signal $u$ induces would have to come from somewhere else in the pipeline. We check this on the main-line data of Table~1 (\texttt{ST\_600.pt}: $600$ queries $\times$ $20$ scenarios, $12{,}000$ rows; \texttt{partial\_corr\_u.py}, \texttt{partial\_corr\_u.json}) by removing successively more of the common structure and recomputing the Pearson correlation of the two brands' residual bids. The signal enters the prompts as one of five demand-level texts (very weak, below normal, normal, elevated, very strong), which defines the cells of the last layer. Bootstrap intervals over queries.

\begin{center}\resizebox{\columnwidth}{!}{\small
\begin{tabular}{lrr}
\toprule
Residualized on & Correlation & $95\%$ CI \\
\midrule
query means (the within-query correlation) & $+0.66$ & $[+0.64,+0.67]$ \\
query means $+$ demand-level means & $+0.16$ & $[+0.14,+0.19]$ \\
(query, demand-level) cell means & $-0.01$ & $[-0.03,+0.01]$ \\
\bottomrule
\end{tabular}
}\end{center}

The demand level explains $88\%$ and $82\%$ of the two brands' within-query bid variance, and once query and shock realization are both held fixed the residual correlation is indistinguishable from zero. The benchmark therefore contains no correlation of its own making: what the mechanism conditions on is the agents' transmission of a common demand signal, the classical source of affiliation (Milgrom and Weber 1982), and the prompt study above measures how much of that signal different agents let through. (The prompt-ablation table above is computed on the pilot sets, and its bid-shape columns pool all four brands; the behavioral statistics in the main text, Sections~7.1 and~7.7, are computed on the two main-line brands in \texttt{ST\_600.pt}: gradient $5.76$, within-query correlation $+0.66$, median bid $37$ cents, 95th percentile $112$; \texttt{behavior\_mainline.py}, \texttt{behavior\_mainline.json}.)

\paragraph{CAN baseline (main text, Table~1).}
We could not find a public implementation of CAN, so every number reported for it comes from our own reimplementation of the mechanism as described in the paper, not from the authors' code.
We reimplement the Conditional Auction Net of Huo, Zheng and Wu (AAAI 2025) as a context-free mechanism on the main line (\texttt{can\_baseline\_fixed.py}, results \texttt{can\_baseline\_fixed\_results.json}). As in the original paper, CAN maps the bids alone to an allocation and a payment. We keep it that way rather than add a context input of our own design, which would make it a different mechanism. The value of context is measured separately, by the contextual-reserve baseline (ctxSPA) and by the fixed-opponent market below, where every context-free mechanism is capped by the no-context ceiling. Each bidder $i$ has a rank-score function $r_i(\mathbf b)=\min_z\max_q[\,w_{zq}\cdot(b_i,-\mathbf b_{-i})+c_{zq}\,]$, a min-max network (six groups of six affine units, nonnegative weights via softplus) over the inputs $(b_i,-\mathbf b_{-i})$, so $\partial r_i/\partial b_i\ge0$ and $\partial r_j/\partial b_i\le0$ hold by construction (the single-crossing conditions, the paper's eq.~(4)). Following the paper's eq.~(5), a learned reserve score $r_0(\mathbf b)=\min_z\max_q[\,w^0_{zq}\cdot(-\mathbf b)+c^0_{zq}\,]$, nonincreasing in every bid, plays the role of an additional bidder: the item goes to $\arg\max_i r_i(\mathbf b)$ if that score is at least $r_0(\mathbf b)$, otherwise it is unsold. The winner pays its critical bid, the smallest report at which it still wins, found by $40$ bisection steps (its own score is nondecreasing and every competing score, including $r_0$, nonincreasing in the report, so the root is unique). Training maximizes expected revenue on the $7{,}200$ training bids of the Table~1 split with a softmax-relaxed allocation over the $n+1$ scores (temperature annealed from $0.3$ to $0.02$ over $3{,}000$ Adam steps, batch $256$, learning rate $2\times10^{-3}$) and gradients through the critical bid by implicit differentiation at interior roots (a payment clipped at zero or at the bid contributes no gradient, as in the original training objective); the checkpoint with the highest validation revenue is kept. The mechanism receives no query context, as in the original paper, and no prior.

Test revenue over five training seeds is $37.43\pm0.07$ cents ($37.43$, $37.44$, $37.48$, $37.46$, $37.32$); four seeds learn a reserve low enough to trade on every test profile (trade rate $100\%$, welfare $54.6$--$54.9$), the fifth trades $85\%$ (welfare $53.2$). An incentive-compatibility probe (100 test profiles, 25 deviations per bidder) finds $0/5{,}000$ violations for every seed. Against the other baselines CAN is clearly ahead: on seed $0$ it earns $+5.1\%$ over the tuned reserve and $+8.0\%$ over IPV-Myerson. \textsc{AffMyerson}'s five-seed revenue on the same test rows is $37.84\pm0.06$ against CAN's $37.43\pm0.07$. The two seed distributions are disjoint (the smallest value the \textsc{AffMyerson} mean and sd admit, $37.74$, exceeds CAN's largest seed, $37.48$), and a Welch test on the seed means gives $t=10.3$, $p<0.001$. The original paper trains CAN on $100{,}000$ samples per market, against $7{,}200$ here.

\paragraph{When the conversation matters to one bidder only (main text, Section~7.3).}
$200$ pairs of conversations ($400$ in all) share the same product words. One member of each pair carries a high-value cue (an upcoming trip, a broken item, a group purchase, a gift with a budget) and the other a low-value cue (writing an article, a school project, an item already bought, a same-name non-product). Each conversation has $20$ rows ($8{,}000$ rows), embedded with the same frozen sentence encoder as the main line, and split $6{:}1{:}3$ by conversation with the main-line split seed ($240/40/120$ conversations; $2{,}400$ test rows, $46\%$ high-cue). Bidder A (``marketplace'') does not read the conversation, and its value is $\mathcal N(30,3^2)$ cents on every row. Bidder B (``specialist'') reads it, with value $\mathcal N(200,25^2)$ on a high cue and $\mathcal N(100,25^2)$ on a low cue, clipped at zero. The two are independent by construction, so this experiment isolates what reading the conversation is worth. It is not a measurement of the affiliation rent, and \textsc{AffMyerson} and IPV-Myerson with the same prior coincide here. The prior is the spectral parameterization of Section~6 fitted on the training rows (five seeds), and CAN is the full model of the paragraph above, trained on the same rows (five seeds). These CAN runs use the corrected trainer described above. One reference quantity: the \emph{no-context ceiling} is Myerson's optimal auction (with ironing) for the empirical marginal bid distributions of the training set, which is the most any mechanism that sees only bids can earn. Code \texttt{ctx\_branch/fixed\_opponent.py}, results \texttt{fixed\_opponent\_results.json}.

\begin{center}\resizebox{\columnwidth}{!}{\small
\begin{tabular}{@{}lrrr@{}}\toprule
Mechanism & Revenue (cents) & High-cue rows & Low-cue rows \\\midrule
Second price (SPA) & $29.93$ & & \\
Tuned-reserve SPA ($r^*{=}163.5$) & $71.26$ & & \\
No-context ceiling (ironed Myerson) & $84.39$ & $155.0$ & $24.6$ \\
CAN (five seeds) & $83.95\pm0.17$ & $153.5$ & $25.1$ \\
Lookahead (same prior) & $100.33\pm0.28$ & & \\
Contextual-reserve SPA & $103.92\pm0.21$ & & \\
IPV-Myerson (same prior) & $104.79\pm0.12$ & & \\
\textsc{AffMyerson} (five seeds) & $104.78\pm0.22$ & $149.7$ & $66.8$ \\\bottomrule

\end{tabular}
}\end{center}
\noindent The fitted prior reads the specialist's mean as $192.5$ on high-cue and $106.3$ on low-cue conversations (true values $200$ and $100$) with conditional scale $31$ (true $25$).

\paragraph{Prompt scheme $\times$ mechanism: conditioning on calculation-type bids.}
To ask what conditioning is worth if bids looked like the calculation-type schemes, the v4.2$+$CoT scheme was rerun at main-line scale ($600$ queries $\times$ $20$ scenarios $\times$ $2$ brands $=24{,}000$ calls; \texttt{run\_full\_cot.py}, data \texttt{ST\_cot.pt}) and passed through the same fitting and auction pipeline (five seeds, split $0$). The bid shape matches the pilot (median $5$ cents, $82\%$ of bids at or below $10$ cents, three values carrying $55\%$). The fitted spectral prior has $\bar\rho=0.17$--$0.19$ on the test set; AffMyerson earns $4.995\pm0.184$ cents, $+2.50\%\pm1.01\%$ over IPV-Myerson and $-7.24\%\pm3.42\%$ against the tuned-reserve SPA ($r^\star=10$), both directions stable across the five seeds (\texttt{cot\_auc\_output.txt}). This is the low-correlation regime of the shuffling scan (main text, Section~7.4), in which the two-tier criterion tells a platform not to condition. Note the two correlation measures differ: the table's Corr.\ is the raw within-query correlation of pilot bids; $\bar\rho$ is the fitted prior's conditional correlation on held-out queries.

We requested access to the Yahoo! Webscope A1 advertiser bidding data, which records competing bids on the same keywords; two months after the request, access had not been granted.
Queries, rather than individual scenarios, are assigned to train, validation,
and test partitions. Reserves are tuned only on the training partition.
Reported fixed-split uncertainty uses five optimization seeds. A separate
five-split audit repeats preprocessing, fitting, reserve tuning, and testing
for independent query partitions. Test-set comparisons use paired bootstrap
resampling of query clusters.

Context-predictability diagnostics on real queries (the $R^2$ figures for the
diagnostic classes and the encoder/feature ablations) are measured with the unconstrained
class under the legacy fitting pipeline; they are an optimistic estimate of what a context mean
could recover and are reported as diagnostics, not properties of the deployed
prior. The designed-data figure they are compared with in the main text ($R^2=0.60$) is the diagonally dominant fit of the four-prior table (Section~S5), the best context mean any class attains there; the unconstrained class itself reaches $R^2\approx0$ on the designed data.

The accompanying code records the exact data file, split seed, optimization
seed, per-split revenue, and paired intervals in machine-readable JSON. Seven
unit tests cover SPD construction, M-matrix feasibility, score monotonicity,
cross-monotonicity, exact reduction to IPV pricing for diagonal $Q$, round-trip
spectral representation, and the critical-payment score equality and bracket;
the lemma-level diagnostics of Section~\ref{supp:stab} (\texttt{check\_z0.py},
\texttt{check\_stability.py}) ship alongside them.

\section{Log-link revenue comparison (controlled, legacy pipeline)}

Same test set, same held-out real bids, same reserve grid; the two blocks
differ only in the prior's link function (both fit by the legacy
diagonally dominant pipeline, which predates the spectral network and budgets partial correlations at roughly $1/(n-1)$, so the rent columns are
\emph{not} comparable to the main line's).

\begin{center}\resizebox{\columnwidth}{!}{\small
\begin{tabular}{lrr}\toprule
Mechanism & Identity link & Log link ($c_0=0.45$) \\\midrule
SPA & 28.39 & 28.39 \\
SPA + tuned reserve ($r^*{=}32$) & 35.62 & 35.62 \\
IPV-Myerson & 33.31 & 36.81 \\
AffMyerson & 36.15 & 37.35 \\\bottomrule
\end{tabular}
}\end{center}

The better-fitting marginal is revenue-positive for both prior-based
mechanisms (the SPA rows are link-free sanity checks); the measured rent
narrows because IPV-Myerson gains more, an effect confounded by the legacy
pipeline's correlation cap. Constructive regularity: the diagonal floor
$c_0=0.45>1/\sigma^{\star2}\approx0.4342$ guarantees $\sigma<\sigma^\star$
for every network output.

\section{Shuffling scan: prior-tracking panel}

Companion panel to the main text's shuffling-scan figure: the fitted mean
correlation $\bar\rho$ of the retrained prior tracks the injected within-query
correlation at every level.

\begin{figure}[h]\centering
\includegraphics[width=0.9\columnwidth]{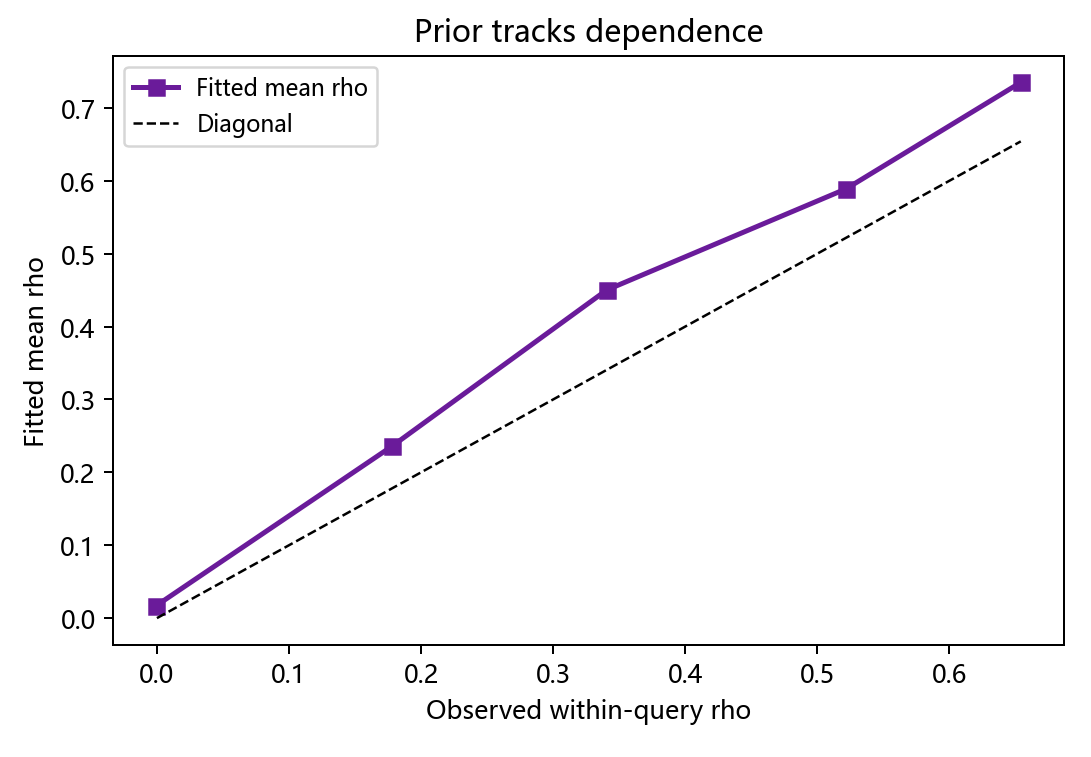}
\end{figure}

\paragraph{Fixed-policy variant (main text, Section~7.4).} The retrained protocol refits the prior and re-tunes the reserve at every shuffling level. The fixed-policy variant keeps the prior fitted on the original data (fitted mean correlation $\bar\rho=0.735$) and the original tuned reserve ($r^\star=32$), and re-pairs only the test set. With no retraining the rent still rises monotonically with the within-query correlation, so the response in the main text is not a retraining artifact; the negative endpoint at $\rho\approx0$ is the cost of conditioning on a correlation that is no longer in the data. Bootstrap intervals over test queries; script \texttt{rho\_scan\_fixed.py}, results \texttt{rho\_fixed\_results.json}.

\begin{center}\resizebox{\columnwidth}{!}{\small
\begin{tabular}{ccrr}
\toprule
Shuffle fraction & Test $\rho$ & Lift vs.\ IPV-Myerson & Lift vs.\ tuned reserve \\
\midrule
$1.00$ & $0.009$ & $-2.2\%$ $[-3.7,-0.5]$ & $-2.2\%$ $[-3.7,-0.8]$ \\
$0.75$ & $0.183$ & $+0.7\%$ $[-0.8,+2.4]$ & $+0.3\%$ $[-1.3,+1.9]$ \\
$0.50$ & $0.365$ & $+4.4\%$ $[+2.8,+6.2]$ & $+3.2\%$ $[+1.6,+4.7]$ \\
$0.25$ & $0.521$ & $+6.9\%$ $[+5.3,+8.8]$ & $+5.0\%$ $[+3.3,+6.6]$ \\
$0.00$ & $0.654$ & $+8.9\%$ $[+7.2,+10.9]$ & $+6.5\%$ $[+4.7,+8.3]$ \\
\bottomrule
\end{tabular}
}\end{center}

\section{Link choice: where each marginal fits better}

\paragraph{Two home turfs: where each marginal fits better (main text, Section~7.2).}
For every (query, brand) cell with at least ten scenario bids we fit a normal and a log-normal by maximum likelihood and compare their negative log-likelihoods on the original bid scale (the log-normal includes the Jacobian term $\sum\log x$; both have two parameters, so the comparison is a likelihood ratio). The table reports the share of cells the normal wins, with $95\%$ binomial half-widths (\texttt{link\_home\_turf.py}, \texttt{link\_home\_turf\_out/results.json}). Overall the log-normal is ahead ($36\%$ normal wins on the designed data, $31\%$ on real queries), but the normal is on par or ahead exactly where revenue concentrates: transactional and niche-premium queries, home-turf queries, and the value and clearance brands whose bids sit in the high-value range with additive rather than multiplicative spread.

\begin{table*}[t]\centering\small
\begin{tabular}{lcc}
\toprule
Cell group (share of cells won by the normal) & Designed ($n{=}4$, $3{,}065$ cells) & Real queries ($2{,}232$ cells) \\
\midrule
All cells & $36\pm2\%$ & $31\pm2\%$ \\
\midrule
Transactional queries & $57\pm4\%$ & $43\pm4\%$ \\
Niche-premium queries & $59\pm6\%$ & --- \\
Home-turf queries (four brands) & $43$--$61\%$ & --- \\
Comparison queries & $21\pm3\%$ & $34\pm4\%$ \\
Informational queries & $22\pm3\%$ & $23\pm4\%$ \\
Adjacent queries & $26\pm3\%$ & $23\pm3\%$ \\
\midrule
Clearance brand (Basecamp) & $48\pm4\%$ & $56\pm4\%$ \\
Value brand (Trail) & $47\pm4\%$ & $45\pm4\%$ \\
Premium brand (Summit) & $27\pm3\%$ & $13\pm3\%$ \\
Flagship brand (Apex) & $23\pm3\%$ & $8\pm2\%$ \\
\bottomrule
\end{tabular}
\end{table*}

\section{Is the network necessary? Convex fits in the same cone}

Real-query benchmark, same split and replay as the channel decomposition
(cents; deterministic fits, network rows are 5-seed means).

\begin{center}\resizebox{\columnwidth}{!}{\small
\begin{tabular}{lr}\toprule
Prior & AffMyerson revenue \\\midrule
SPA + tuned reserve ($r^*{=}25$) & 41.54 \\
Embedding-context network & 43.48 \\
Global closed form + cone projection (paper) & 43.65 \\
Global $\mathrm{MTP}_2$-constrained MLE (convex) & 43.65 \\
Diagonal $+$ rank-one single factor & 43.63 \\
Intent-bucketed $\mathrm{MTP}_2$-MLE (4 buckets) & 44.29 \\
Intent-feature network & 44.33 \\\bottomrule
\end{tabular}
}\end{center}

The convex MLE and the clip-projection agree to $0.001$; the bucketed
convex fit ties the intent-feature network; the single-factor prior loses
only $0.02$ to the full M-matrix fit here. On the \emph{designed} four-brand
data the ordering changes (seed-0 configuration, split $0$, cents): unconstrained network repaired into the cone $49.29$;
projected global $48.33$; convex global MLE $48.17$; SPA + tuned reserve
$47.09$ ($r^\star=30$); diagonal-plus-rank-one $46.31$ (\texttt{convex\_compare\allowbreak\_n4b.py}, \texttt{convex\_compare\allowbreak\_n4b.json}; the network row is the unconstrained precision network of \texttt{proj\_cost.py} after projection into the cone, \texttt{proj\_cost\_out.json}). The projected global fit and the convex MLE differ here because clipping the sample precision matrix distorts the marginal variances (its implied IPV revenue is $43.64$ against $45.95$ for the MLE, which preserves them), while their conditioning revenues stay within $0.16$ cents. Where context is predictable the
repaired unconstrained network is worth $+0.96$ over the best convex fit, and one Gaussian factor
is not enough, the applicability boundary reported in the main text.

\section{Collusion sensitivity and the conditioning clip}\label{sec:clip}

Single-shot DSIC covers unilateral deviations only. The main text reports that joint shading by the winner's opponents moves the winner's payment more under conditioning than under IPV pricing, because the same channel that earns the rent transmits coordinated misreports. The exposure is tunable. Define the \emph{clipped} conditional mean by capping the shift that competitors' reports induce,
\[
  m^{(c)}_{i}=\mu_i+\operatorname{clip}\Bigl(\textstyle\sum_{j\ne i}\beta_{ij}(s_j-\mu_j),\ -c\,\sigma_{i\mid-i},\ c\,\sigma_{i\mid-i}\Bigr),
\]
and run the mechanism with $\psi^{(c)}_i=s_i-\sigma_{i\mid-i}R\bigl((s_i-m^{(c)}_i)/\sigma_{i\mid-i}\bigr)$. The clip is monotone in every $s_j$, so own-monotonicity and cross-monotonicity (main text, Lemmas~5.1--5.2) are preserved and the mechanism stays EPIC and EPIR for every $c\in[0,\infty]$; $c=\infty$ is \textsc{AffMyerson}, $c=0$ prices with the conditional scale but no conditional shift. Seed-$0$ spectral prior, held-out $n{=}2$ bids, script \texttt{spectral\_unify.py} (results in \texttt{spectral\_unify\_results.json}); ``drop'' is the winner's payment reduction when all opponents shade their reports by $20\%$; welfare is relative to first best; allocations are unchanged by the shading in every row; single seed, hence the small offsets from the five-seed means of the main text.

\begin{table*}[t]\centering\footnotesize\setlength{\tabcolsep}{4pt}
\begin{tabular}{lccccc}\toprule
Mechanism & Revenue & Rent vs.\ IPV & Share of rent & Welfare & Drop, $20\%$ shade \\\midrule
IPV-Myerson (no conditioning) & $34.83$ & -- & -- & $84.2\%$ & $10.1\%$ \\
$c=0$ (scale only, no shift) & $35.16$ & $+1.0\%$ & $11\%$ & $87.6\%$ & $11.0\%$ \\
$c=0.25$ & $36.56$ & $+5.0\%$ & $56\%$ & $90.8\%$ & $11.4\%$ \\
$c=0.5$ & $37.21$ & $+6.8\%$ & $77\%$ & $92.7\%$ & $12.1\%$ \\
$c=1$ & $37.66$ & $+8.1\%$ & $91\%$ & $93.9\%$ & $13.1\%$ \\
$c=2$ & $37.87$ & $+8.7\%$ & $98\%$ & $94.0\%$ & $13.8\%$ \\
$c=\infty$ (full conditioning) & $37.93$ & $+8.9\%$ & $100\%$ & $93.7\%$ & $14.5\%$ \\
\bottomrule
\end{tabular}
\end{table*}

\noindent Full conditioning earns $+8.9\%$ over IPV and loses $14.5\%$ of the winner's payment to a $20\%$ joint shade, against $10.1\%$ for IPV; the extra exposure is $4.4$ points. At $c=0.5$ the mechanism keeps $77\%$ of the rent while the extra exposure falls to $2.0$ points; at $c=0.25$, $56\%$ of the rent for $1.3$ points. The trade-off is monotone in $c$, so a platform can set the clip from its own estimate of collusion risk; nothing else in the pipeline changes. Every clipped variant is EPIC/EPIR by the main text's Theorem~5.7; only $c=\infty$ is the revenue-optimal rule of Theorem~5.5 for a Gaussian prior. (Numbers here are the spectral main line and supersede any legacy-pipeline figures.)

\section{Prior-poisoning attack and mitigation (full numbers)}

Shading one advertiser's bid history by $\delta=10\%/20\%$ before the refit
(period-one cost ignored) raises that advertiser's next-period utility by
$+14\%/+37\%$, while platform revenue falls by $0.6\%/3.0\%$ (the IPV baseline
loses $1.6\%$ under the same attack). Mitigation by prior averaging: replacing
the current prior with a convex combination of the current and previous-period
priors cuts the $\delta=20\%$ utility gain to $+17\%$ at $\lambda=\tfrac12$ and
$+8\%$ at $\lambda=\tfrac14$, at no revenue cost when no attack is present.
Full dynamic-IC treatment remains open; single-shot EPIC is unaffected.

\section{Complete rerun summary}\label{supp:rerun}

\begin{table}[h]\centering
\caption{Main fixed-split revenue, five optimization seeds.}
\begin{tabular}{lr}\toprule
Mechanism & Revenue (cents)\\\midrule
SPA & 28.391\\
Contextual-reserve SPA & $32.979\pm0.179$\\
IPV-Myerson & $34.661\pm0.188$\\
Ronen lookahead & $35.667\pm0.075$\\
SPA + tuned reserve & 35.621\\
AffMyerson & $\mathbf{37.838\pm0.056}$\\\bottomrule
\end{tabular}
\end{table}

\begin{table}[h]\centering
\caption{Independent query-split reruns with the spectral M-matrix prior.}
\begin{tabular}{crr}\toprule
Split & Lift vs.\ tuned reserve & Lift vs.\ IPV\\\midrule
0 & $+5.80\%$ & $+8.14\%$\\
1 & $+6.78\%$ & $+7.94\%$\\
2 & $+5.80\%$ & $+9.59\%$\\
3 & $+4.60\%$ & $+8.91\%$\\
4 & $+5.97\%$ & $+9.80\%$\\\midrule
Mean$\pm$sd & $+5.79\%\pm0.78\%$ & $+8.88\%\pm0.83\%$\\\bottomrule
\end{tabular}
\end{table}

On the full $n=4$ market, AffMyerson earns $49.047\pm0.431$ cents,
$4.16\%\pm0.92\%$ above the tuned reserve and $6.17\%\pm0.37\%$ above matched
IPV. Sweeping the share of special (niche/home-turf) queries from $0$ to $23\%$ moves the rent between $+5.5\%$ and $+6.8\%$ and the margin over the tuned reserve between $+4.3\%$ and $+7.6\%$, all confidence intervals positive. On the $n=2$ main line AffMyerson also trades more often than the tuned-reserve SPA ($89.8\%$ vs.\ $83.3\%$ of auctions) at equal welfare ($53.8$). On real-user queries with LLM-agent valuations, it earns
$43.480\pm0.146$ cents, with corresponding lifts of
$4.67\%\pm0.35\%$ and $4.48\%\pm0.07\%$.

\paragraph{Brand submarkets (main text, Section~7.5).} Because each brand bids without seeing the others, every subset of the four brands is a valid market. The table lists all eleven submarkets with the spectral prior (seed $0$; bootstrap intervals over test queries; script \texttt{rho\_scan\_b.py}, results \texttt{n\_sweep\_out\_b/results\_spectral.json}). Fitted $\bar\rho$ alone does not rank the markets: the three lowest rents ($+0.4$, $+0.5$, $+1.5\%$, none significant) are the pairings of a high-end brand with the clearance brand, whose near-constant bids carry no usable signal for conditioning, while the same clearance-containing markets show the largest margins over the tuned reserve because a single reserve fits them poorly. One three-brand market loses to the tuned reserve ($-1.5\%$ $[-2.6,-0.2]$). The across-submarket median rent is $+6.3/+6.5/+6.4\%$ for $n=2/3/4$; a symmetric synthetic market ($\mu=50$, $\sigma=15$ cents, equicorrelation $\rho=0.6$, $5{,}000$ auctions per $n$; \texttt{mech\_aff\_n.py}) declines instead, $+4.0/+3.3/+3.1/+2.2\%$ for $n=2/3/4/5$, so in the stratified market the conditioning information that an added competitor brings offsets the dilution of the rent.

\begin{table*}[t]\centering\small
\begin{tabular}{clrr}
\toprule
$n$ & Market & Rent vs.\ IPV-Myerson & Lift vs.\ tuned reserve \\
\midrule
2 & Flagship$+$Premium & $+5.6\%$ $[+2.5,+9.0]$ & $+1.6\%$ $[-2.9,+6.8]$ \\
2 & Flagship$+$Value & $+7.1\%$ $[+5.2,+9.0]$ & $+1.3\%$ $[-1.0,+3.2]$ \\
2 & Flagship$+$Clearance & $+0.5\%$ $[-2.0,+3.1]$ & $+22.2\%$ $[+13.5,+31.8]$ \\
2 & Premium$+$Value & $+8.4\%$ $[+6.7,+10.2]$ & $+10.5\%$ $[+7.8,+13.2]$ \\
2 & Premium$+$Clearance & $+1.5\%$ $[-0.3,+3.2]$ & $+37.9\%$ $[+31.2,+44.9]$ \\
2 & Value$+$Clearance & $+8.4\%$ $[+7.1,+9.7]$ & $+12.4\%$ $[+10.8,+14.2]$ \\
3 & Flagship$+$Premium$+$Value & $+6.4\%$ $[+4.3,+8.8]$ & $-1.5\%$ $[-2.6,-0.2]$ \\
3 & Flagship$+$Premium$+$Clearance & $+0.4\%$ $[-1.5,+2.5]$ & $+3.4\%$ $[+0.8,+6.4]$ \\
3 & Flagship$+$Value$+$Clearance & $+8.9\%$ $[+7.4,+10.5]$ & $+19.4\%$ $[+15.2,+23.8]$ \\
3 & Premium$+$Value$+$Clearance & $+6.7\%$ $[+4.9,+8.5]$ & $+15.4\%$ $[+11.9,+19.3]$ \\
4 & All four & $+6.4\%$ $[+4.8,+7.9]$ & $+4.7\%$ $[+3.5,+5.9]$ \\
\bottomrule
\end{tabular}
\end{table*}

\paragraph{Parameterization audit.}
On $n=4$, relaxed precision learning plus feasibility repair clips
$24$--$27\%$ of off-diagonal entries, removes $16$--$21\%$ of their
$L_1$ mass, and shifts the diagonal of $36$--$43\%$ of matrices; its
reserve-relative effect changes sign. The diagonally dominant M-matrix
ablation is stable but fits the flagship--premium pair correlation near $0.23$ (all-pairs mean $0.22$).
The spectral model fits $0.58\pm0.05$ for the same pair (all-pairs mean $0.47$; four-prior table, Section~S5) and gives stable revenue effects.

\paragraph{Distributional misspecification audit.}
On the main $n=2$ split, we refit the same spectral M-matrix network after
monotone transforms and compare test NLL on the original bid scale using the
exact Jacobian. Across five optimization seeds, the original-scale Gaussian
has NLL $9.410\pm0.051$ and whitened-residual standard deviation, skewness,
and kurtosis $1.410\pm0.022$, $0.846\pm0.125$, and $7.996\pm0.726$.
The $\log(s+0.1)$ transform has NLL $8.435\pm0.016$, an improvement of
$0.976\pm0.057$ nats, with residual diagnostics $1.088\pm0.032$,
$-0.066\pm0.022$, and $3.124\pm0.079$. This is evidence against treating the
Gaussian as the literal agent-value distribution. We retain it as a working
prior because the closed-form regularity and incentive proofs are for the
original-scale Gaussian; the transformed fit is a diagnostic, not an auction
baseline. Accordingly, the experiment supports structural truthfulness and
replay revenue under misspecification, not the correct-prior optimality premise.

\paragraph{Latency.}
With 60 bisection iterations, the final reference pipeline (spectral network,
tensor transfer, and pricing; sentence encoder excluded) has p50/p95 latency
$0.308/0.450$ ms for $n=2$, $0.676/0.946$ ms for $n=4$, and
$0.887/1.195$ ms for $n=5$ on the reported AMD CPU.

\paragraph{Unseen-shock audit.}
The prior and reserve are trained only on the central $80\%$ of latent demand
shocks; evaluation uses held-out queries in the unseen lower and upper $10\%$
tails. The latent shock is used only to define this audit and is not supplied
to the network or mechanism. Across five optimization seeds, combined-tail
revenue is $37.046\pm0.271$ cents, a $9.52\%\pm0.49\%$ lift over matched IPV
and $14.08\%\pm0.84\%$ over the training-tuned reserve. Seed-0 paired
query-cluster intervals are $[6.21,12.32]\%$ and $[8.89,16.82]\%$.
As a deliberately optimistic diagnostic, we additionally select one constant
reserve directly on the combined tail test sample. This post-hoc oracle
reserve is not deployable and has an information advantage over AffMyerson;
nevertheless, AffMyerson exceeds it by $8.54\%\pm0.80\%$ across seeds, with a
seed-0 paired query-cluster interval of $[3.26,11.81]\%$. Thus the combined-tail
result is not explained only by carrying a central-regime reserve into the
tails.
If the diagnostic is given the hidden tail label and may tune a different
reserve separately in each regime, those two test-tuned oracles exceed
AffMyerson by $31.21\%\pm2.93\%$ in the low tail and
$3.29\%\pm0.50\%$ in the high tail. This is not a deployable shock-blind
baseline, but it identifies a remaining limitation: the present context model
cannot adapt to an unobserved regime label.

\section{Use of generative AI}

Generative AI (Claude, Anthropic, via the Claude desktop app) was used to edit prose, to write and review code, to check proofs and propose fixes for gaps the authors then verified, and to help refine experimental designs proposed by the authors. All mathematical statements, citations, code changes, and numerical results were checked by the authors. A short note on which research stages received AI assistance, with the tool version, will accompany the released code. The LLM agents that act as bidders in the experiments are part of the experimental setup, not writing assistance, and are documented with their prompts and sampling settings.

\section{Real-query prediction and channel diagnostics}

This section provides the per-source and channel decompositions supporting Section~7.6 of the main paper.

\paragraph{Construction of the real-query benchmark (main text, Section~7.6).} Queries come from two public corpora of genuine user queries, the Amazon ESCI shopping-query examples and the MS MARCO query set. Each corpus is first screened by outdoor-gear keywords, then ranked by cosine similarity of its MiniLM embedding to five seed phrases (``hiking and camping gear'', ``outdoor equipment for trekking'', ``mountaineering and climbing equipment'', ``ski and snowboard gear'', ``camping tent sleeping bag backpack''), keeping the top $2{,}000$ ESCI and $3{,}200$ MARCO queries. Two \texttt{gpt-4o-mini} passes follow. The first labels each query with one of the four intent tiers used for the designed queries (transactional, comparison, informational, adjacent), with the same definitions. The second is a topic gate that answers yes or no to whether the query is genuinely about outdoor recreation gear or activities, and rejects gym equipment, toys, baby gear, school bags, grills, RV dealers, road-cycling brands, fashion items, events and stores, and anything that merely contains outdoor-sounding words. From the queries that pass the gate we draw up to $150$ per tier, filling transactional and comparison tiers from ESCI first and informational and adjacent tiers from MARCO first, which yields the $580$ queries ($274$ ESCI, $306$ MARCO) of the main text. Neither pass sees any bid or revenue. Script \texttt{build\_real\_queries.py}; both prompts are in the script and the released artifact.

\paragraph{Per-source decomposition.} The two sources are independent sub-benchmarks; the split localizes the prediction failure: $R^2$ is mildly positive on MARCO ($+0.13$) but \emph{negative} on ESCI ($-0.45$); within concrete shopping queries the value differences ride on product-tier and budget details invisible to sentence embeddings, so the model's variation is spurious; a modern encoder (bge-base, C-Pack) doubles overall $R^2$ to $0.21$ without flipping ESCI; a $4$-dimensional intent label alone reaches $0.18$; a split-half noise ceiling puts attainable $R^2$ at $\approx0.81$ (split-half estimate). Crucially, the rent persists \emph{on exactly the source where text prediction is counterproductive}.

\paragraph{Channel decomposition.} A $2\times2$ ablation (context on/off $\times$ conditioning on/off, same held-out bids) makes ``the rent flows through conditioning'' quantitative: with a \emph{global} prior (zero text prediction) conditioning alone is worth $+3.03$ cents, \emph{more} than the entire rent ($+1.86\pm0.03$), and already beats the tuned-reserve SPA by $+5.1\%$; the embedding context channel adds \emph{nothing} ($-0.17\pm0.14$), while the $4$-dim intent feature adds $+0.68\pm0.03$. The same $2\times2$ on the \emph{designed} data completes the contrast: conditioning is worth $+4.48$ cents there (vs.\ rent $+3.18\pm0.16$) while the context channel is worth $\approx+1.9$. Conditioning is the engine in both regimes. Revenue is insensitive to $\mu$-accuracy beyond a coarse level: the engine is conditioning on competitors' reports (main-paper Theorem~5.7).

\end{document}